\documentclass{article}

\usepackage{arxiv}

\usepackage[utf8]{inputenc} 
\usepackage[T1]{fontenc}    
\usepackage[hidelinks]{hyperref}       
\usepackage{url}            
\usepackage{booktabs}       
\usepackage{amsfonts}       
\usepackage{nicefrac}       
\usepackage{microtype}      
\usepackage{lipsum}
\usepackage{natbib}
\usepackage{caption} 
\usepackage{amssymb} 
\usepackage{amsmath} 
\usepackage{mathrsfs} 
\usepackage{graphicx}
\usepackage{enumerate}
\usepackage{url} 
\usepackage[amsmath,amsthm,thmmarks]{ntheorem}
\usepackage{bbm} 

\usepackage{subcaption} 
\usepackage{appendix}
\usepackage[ruled,vlined]{algorithm2e} 

\newtheorem{theorem}{Theorem}
\newtheorem{lemma}{Lemma}
\newtheorem{assumption}{Assumption}
\newtheorem{remark}{Remark}
\newtheorem{corollary}{Corollary}
\newtheorem{definition}{Definition}

\title{Stochastic Approximation Cut Algorithm for Inference in Modularized Bayesian Models}

\author{
  Yang Liu\thanks{Correspondence: 
    Yang Liu, MRC Biostatistics Unit, University of Cambridge, Robinson Way, CB2 0SR} \\
  MRC Biostatistics Unit\\
  University of Cambridge\\
  Cambridge, UK \\
  \texttt{yang.liu@mrc-bsu.cam.ac.uk} \\
   \And
 Robert J.B. Goudie \\
  MRC Biostatistics Unit\\
  University of Cambridge\\
  Cambridge, UK \\
  \texttt{robert.goudie@mrc-bsu.cam.ac.uk} \\
}

\begin{document}
\maketitle

\begin{abstract}
Bayesian modelling enables us to accommodate complex forms of data and make a comprehensive inference, but the effect of partial misspecification of the model is a concern. One approach in this setting is to modularize the model, and prevent feedback from suspect modules, using a cut model. After observing data, this leads to the cut distribution which normally does not have a closed-form. Previous studies have proposed algorithms to sample from this distribution, but these algorithms have unclear theoretical convergence properties. To address this, we propose a new algorithm called the Stochastic Approximation Cut algorithm (SACut) as an alternative. The algorithm is divided into two parallel chains. The main chain targets an approximation to the cut distribution; the auxiliary chain is used to form an adaptive proposal distribution for the main chain. We prove convergence of the samples drawn by the proposed algorithm and present the exact limit. Although SACut is biased, since the main chain does not target the exact cut distribution, we prove this bias can be reduced geometrically by increasing a user-chosen tuning parameter. In addition, parallel computing can be easily adopted for SACut, which greatly reduces computation time.
\end{abstract}


\section{Introduction}
\label{sec:intro}
Bayesian models mathematically formulate our beliefs about the data and parameter. Such models are often highly structured models that represent strong assumptions. Many of the desirable properties of Bayesian inference require the model to be correctly specified. We say a set of models $f(x|\theta)$, where $\theta\in\Theta$, are misspecified if there is no $\theta_0\in\Theta$ such that data $X$ is independently and identically generated from $f(x|\theta_0)$ \citep{WALKER20131621}. In practice, models will inevitably fall short of covering every nuance of the truth. One popular approach when a model is misspecified is fractional (or power) likelihood. This can be used in both classical \citep[e.g.,][]{doi:10.1002/sim.2129, doi:10.1080/13658810802672469, liu2018geographically} and Bayesian \citep[e.g.,][]{miller2018robust, bhattacharya2019bayesian} frameworks. However, this method treats all of the model as equally misspecified.

We consider the situation when the assumptions of the model are thought to partially hold: specifically, we assume that one distinct component \citep[or module in the terminology of][]{liu2009modularization} is thought to be incorrectly specified, whereas the other component is correctly specified. In standard Bayesian inference, these distinct modules are linked by Bayes' theorem. Unfortunately, this means the reliability of the whole model may be affected even if only one component is incorrectly specified. To address this, in this paper we adopt the idea of ``cutting feedback'' \citep{lunn2009bugs, liu2009modularization, Plummer2015, jacob2017better, 2017arXiv170803625J} which modifies the links between modules so that estimation of non-suspect modules is unaffected by information from suspect modules. This idea has been used in a broad range of applications including the study of population
pharmacokinetic/pharmacodynamic (PK/PD) models \citep{lunn2009combining}, analysis of computer models \citep{liu2009modularization}, Bayesian estimation of causal effects with propensity scores \citep{CuttingFeedbackinBayesianRegressionAdjustmentforthePropensityScore, doi:10.1080/00031305.2015.1111260} and Bayesian analysis of health effect of air pollution \citep{BLANGIARDO2011379}.

Consider the generic two module model with observable quantities (data) $Y$ and $Z$ and parameters $\theta$ and $\varphi$, shown in the directed acyclic graph (DAG) in Figure \ref{F1}. The joint distribution is
\[
p(Y,Z,\theta,\varphi)=p(Y|\theta,\varphi)p(Z|\varphi)p(\theta)p(\varphi),
\]
and the standard Bayesian posterior, given observations of $Y$ and $Z$, is
\[
p(\theta,\varphi|Y,Z)=p(\theta|Y,\varphi)p(\varphi|Y,Z)=\frac{p(Y|\theta,\varphi)p(\theta)}{p(Y|\varphi)}\frac{p(Y|\varphi)p(Z|\varphi)p(\varphi)}{p(Y,Z)}.
\]

\begin{figure}
\center
  \includegraphics[scale=0.4]{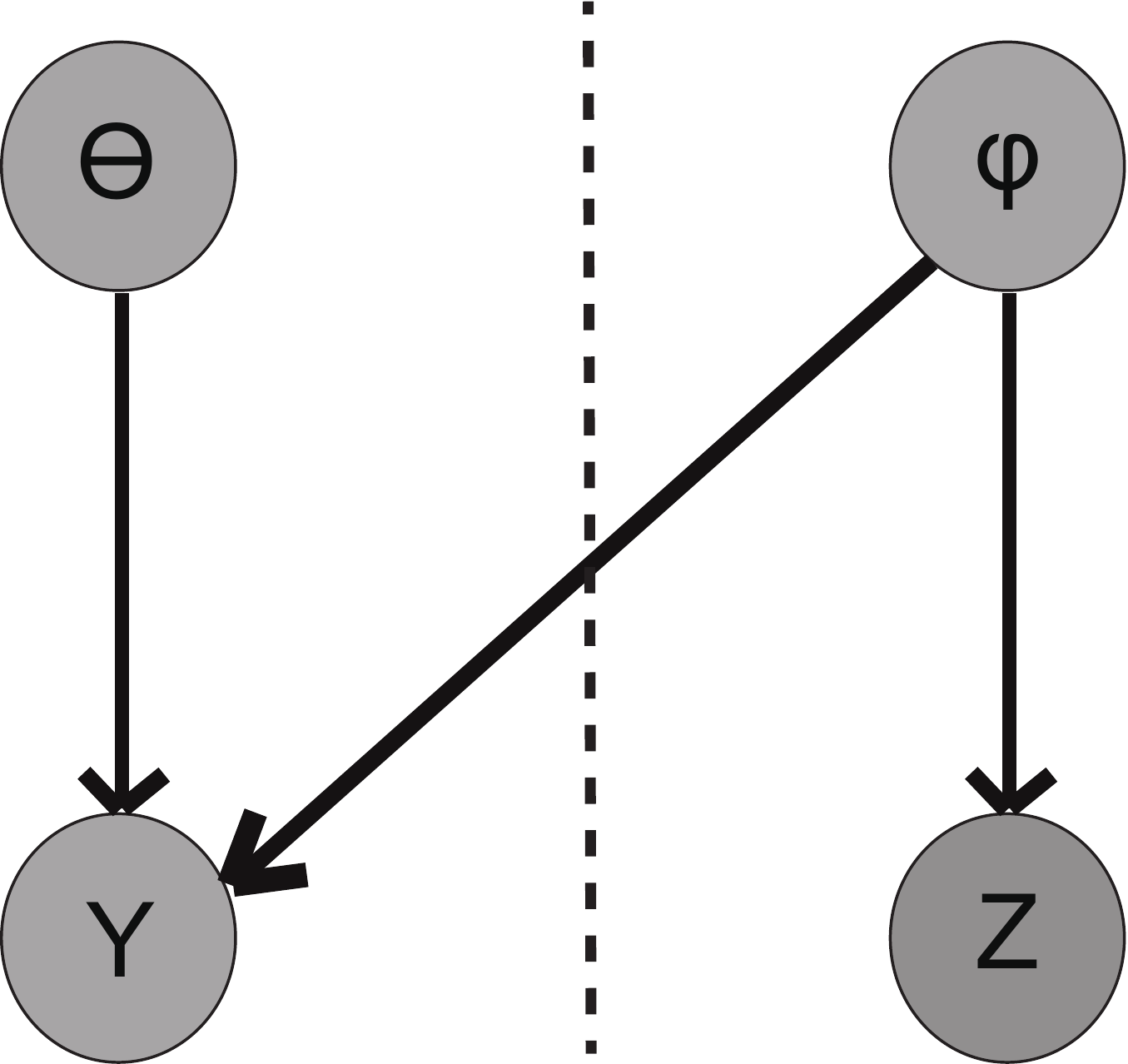}
\caption{DAG representation of a generic two module model. The two modules are separated by a dashed line.}
\label{F1}       
\end{figure}

Suppose we are confident that the relationship between $\varphi$ and $Z$ is correctly specified but not confident about the relationship between $\varphi$ and $Y$. To prevent this possible misspecification affecting estimation of $\varphi$, we can ``cut'' feedback by replacing $p(\varphi|Y,Z)$ in the standard posterior with $p(\varphi|Z)$, making the assumption that $\varphi$ should be solely estimated by $Z$,
\begin{equation}
p_{cut}(\theta,\varphi) :=p(\theta|Y,\varphi)p(\varphi|Z)=\frac{p(Y|\theta,\varphi)p(\theta)}{p(Y|\varphi)}\frac{p(Z|\varphi)p(\varphi)}{p(Z)}.
\label{E3}
\end{equation}
We call \eqref{E3} the ``cut distribution''. The basic idea of cutting feedback is to allow information to ``flow'' in the direction of the directed edge, but not in the reverse direction (i.e. a ``valve'' is added to the directed edge).

Sampling directly from $p_{cut}(\theta,\varphi)$ is difficult because the marginal likelihood $p(Y|\varphi)=\int p(Y|\theta,\varphi)p(\theta)d\theta$ depends on a parameter of interest $\varphi$ and is not usually analytically tractable, except in the simple case when $p(\theta)$ is conditionally conjugate to $p(Y|\theta,\varphi)$, which we do not wish to assume. This intractable marginal likelihood is a conditional posterior normalizing constant: it is the normalizing function for the posterior distribution $p(\theta|Y,\varphi)$, conditional on $\varphi$, of a parameter $\theta$ of interest: 
\begin{equation}
p(\theta|Y, \varphi) = \frac{p(Y, \theta| \varphi)}{p(Y | \varphi)}.
\label{E4}
\end{equation}
This differs importantly to intractable likelihood normalizing constants, as discussed in the doubly intractable literature \citep[e.g.,][]{doi:10.1080/01621459.2018.1448824}, in which the normalizing function $H(\varphi) = \int h(Y | \varphi) dY$ for the likelihood is intractable.
\[
p(Y | \varphi) = \frac{h(Y | \varphi)}{H(\varphi)}.
\]
The normalizing function $H(\varphi)$ is obtained by marginalizing the likelihood, with respect to the observable quantity $Y$, in contrast to the normalizing function $p(Y|\varphi)$, which is obtained by marginalizing likelihood $p(Y, \theta | \varphi)$ with respect to a parameter $\theta$ of interest. This difference means that standard methods for doubly intractable problems \citep[e.g.,][]{10.2307/20441294, Murray:2006:MDD:3020419.3020463, doi:10.1080/00949650902882162, doi:10.1080/01621459.2015.1009072}, which introduce an auxiliary variable, with the same distribution (or proposal distribution) as the distribution of the \textit{a posteriori} observed and fixed $Y$ to cancel the intractable normalizing function shared by them, do not directly apply to \eqref{E4}.

A simple algorithm that aims to draw samples from $p_{cut}(\theta,\varphi)$ is implemented in WinBUGS \citep{lunn2009bugs}. It is a Gibbs-style sampler that involves updating $\theta$ and $\varphi$ with a pair of transition kernels $q(\theta^\prime|\theta,\varphi^\prime)$ and $q(\varphi^\prime|\varphi)$ that satisfy detailed balance with $p(\theta|Y,\varphi^\prime)$ and $p(\varphi|Z)$ respectively. However, the chain constructed by the WinBUGS algorithm may not have the cut distribution as its stationary distribution \citep{Plummer2015} since
\[
\int p_{cut}(\theta,\varphi)q(\theta^\prime|\theta,\varphi^\prime)q(\varphi^\prime|\varphi)d\theta d\varphi=w(\theta^\prime,\varphi^\prime)p_{cut}(\theta^\prime,\varphi^\prime),
\] 
where the weight function $w$ is
\[
w(\theta^\prime,\varphi^\prime)=\int \frac{p(\theta|Y,\varphi)}{p(\theta|Y,\varphi^\prime)}q(\varphi|\varphi^\prime)q(\theta|\theta^\prime,\varphi^\prime)d\theta d\varphi.
\]

The WinBUGS algorithm is inexact since $w(\theta^\prime,\varphi^\prime)\neq 1$, except in the simple case (conditionally-conjugate) when it is possible to draw exact Gibbs updates from $p(\theta^\prime|Y,\varphi^\prime)$. \cite{Plummer2015} proposed two algorithms that address this problem by satisfying $w(\theta^\prime,\varphi^\prime)=1$ approximately. One is a nested MCMC algorithm, which updates $\theta$ from $p(\theta^\prime|Y,\varphi^\prime)$ by running a separate internal Markov chain with transition kernel $q^\ast(\theta^\prime|\theta,\varphi^\prime)$ satisfying detailed balance with the target distribution  $p(\theta|Y,\varphi^\prime)$. The other is a linear path algorithm, which decomposes the complete MCMC move from $(\theta,\varphi)$ to $(\theta^\prime,\varphi^\prime)$ into a series of substeps along a linear path from $\varphi$ to $\varphi^\prime$ and drawing a new $\theta$ at each substep. However, these methods require either the length of the internal chain or the number of substeps to go to infinity, meaning that in practice, these algorithms will not necessarily converge to $p_{cut}$.

In this article, we propose a new sampling algorithm for $p_{cut}(\theta,\varphi)$, called the Stochastic Approximation Cut Algorithm (SACut). Since $\varphi$ can be easily sampled from tractable part $p(\varphi|Z)$, our algorithm aims to sample $\theta$ from the intractable part $p(\theta|Y,\varphi)$. Our algorithm is divided into two chains that are run in parallel: the main chain that approximately targets $p_{cut}(\theta,\varphi)$; and an auxiliary chain that is used to form a proposal distribution for $\theta|\varphi$ in the main chain (see Figure \ref{F3} (b)). The auxiliary chain uses Stochastic Approximation Monte Carlo (SAMC) \citep{doi:10.1198/016214506000001202} to approximate the intractable marginal likelihood $p(Y|\varphi)$ and subsequently form a discrete set of distribution $p(\theta|Y,\varphi)$ for each $\varphi\in\Phi_0=\{\varphi_0^{(1)},...,\varphi_0^{(m)}\}$, a set of pre-selected auxiliary parameters (see Figure \ref{F3} (a)).

The basic ``naive'' form of our algorithm has convergence in distribution, but stronger convergence properties can be obtained by building a proposal distribution $p_n^{(\kappa)}(\theta|Y,\varphi)$ to target an approximation $p^{(\kappa)}(\theta|Y,\varphi)$ instead of the true distribution $p(\theta|Y,\varphi)$ (see Figure \ref{F3} (c-d)). We prove a weak law of large numbers for the samples $\{(\theta_n,\varphi_n)\}_{n=1}^N$ drawn from the main chain. We also prove that the bias due to targeting $p^{(\kappa)}(\theta|Y,\varphi)$ can be controlled by the precision parameter $\kappa$, and that the bias decreases geometrically as $\kappa$ increases. Our algorithm is inspired by the adaptive exchange algorithm \citep{doi:10.1080/01621459.2015.1009072}, but replaces the exchange step with a direct proposal distribution for $\theta$ given $\varphi$ in the main chain.

\begin{figure}[tp] 
\centering 
\includegraphics[origin=c,scale=0.2]{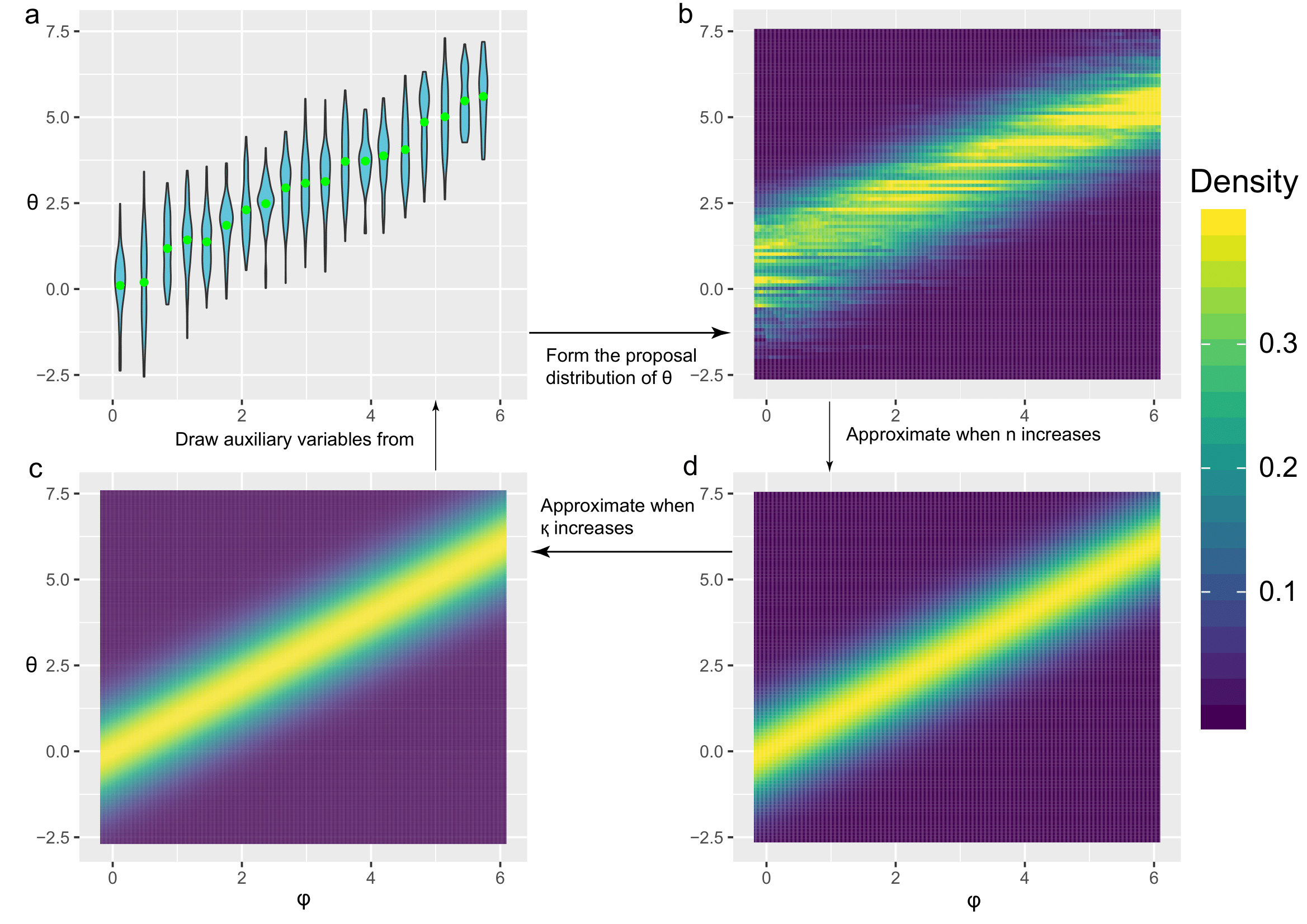}
\caption[Relationship between $p(\theta|Y,\varphi_0^{(i)})$, $p(\theta|Y,\varphi)$, $p^{(\kappa)}(\theta|Y,\varphi)$ and $p_n^{(\kappa)}(\theta|Y,\varphi)$.]{Relationship between $p(\theta|Y,\varphi_0^{(i)})$, $p(\theta|Y,\varphi)$, $p^{(\kappa)}(\theta|Y,\varphi)$ and $p_n^{(\kappa)}(\theta|Y,\varphi)$. This is a toy example when the conditional distribution of $\theta$, given $Y=1$ and $\varphi$, is $\text{N}(\varphi, Y^2)$. Samples of the auxiliary variable $\tilde{\theta}$ are drawn from a mixture of discretized densities $p(\theta|Y,\varphi_0^{(i)})$, $i=1,...,m$, shown in the violin plot in (a), with the green dots showing the median of each component (see Section~\ref{S2.1}). Then $p_n^{(\kappa)}(\theta|Y,\varphi)$, shown in (b), is formed by using these auxiliary variables (see Section~\ref{S2.2}). Lemma \ref{lem1} (Section~\ref{S2.3}) shows that $p_n^{(\kappa)}(\theta|Y,\varphi)$ converges to $p^{(\kappa)}(\theta|Y,\varphi)$, which is shown in (d), while \eqref{E18} shows that $p^{(\kappa)}(\theta|Y,\varphi)$ converges to the original density $p(\theta|Y,\varphi)$, shown in (c).} 
\label{F3}
\end{figure}

\section{Main Result}
\label{sec:main}
Let the product space $\Theta\times\Phi$ be the supports of $\theta$ and $\varphi$ under $p_{cut}$. We assume the following throughout for simplicity.
\begin{assumption}
\label{assu1}
(a) $\Theta$ and $\Phi$ are compact, (b) $p_{cut}$ is continuous with respect to $\theta$ and $\varphi$ over $\Theta\times\Phi$.
\end{assumption}
Assumption \ref{assu1}(a) is restrictive, but is commonly assumed in the study of adaptive Markov chains \citep{haario2001}. In most applied settings it is reasonable to choose a prior for $\theta$ and $\varphi$ with support on only a compact domain, making the domain of the cut distribution compact without any alteration to the likelihood. Note that Assumption \ref{assu1} implies that $p_{cut}$ is bounded over $\Theta\times\Phi$. From now on, define a probability space $(\Omega,\mathcal{F},\mathbb{P})$. Denote Lebesgue measure $\mu$ on $\Theta$ and $\Phi$ and let $P_{cut}$ be the measure on $\Theta\times\Phi$ defined by its density $p_{cut}$.

In following sections, we describe the construction of the algorithm. The naive version of our algorithm builds a discrete proposal distribution for $\theta$, based on \cite{doi:10.1080/01621459.2015.1009072}. Note that, in \cite{doi:10.1080/01621459.2015.1009072}, this proposal distribution only draws auxiliary variables that are discarded once the parameter of interest is drawn, and so strong convergence results for samples drawn by this probability distribution are not needed by \cite{doi:10.1080/01621459.2015.1009072}. This does not always apply to our problem since $\theta$ is the parameter of interest and its parameter space can be continuous. The naive version does not allow us to prove stronger convergence and theoretical properties, so we apply a simple function approximation technique with a specially-designed partition of the parameter space that enables a straightforward implementation. Although this approximation leads to bias, we show that it can be controlled.

\subsection{Naive Stochastic Approximation Cut Algorithm}
\label{S2.1}
To introduce ideas that we will use in Section \ref{S2.3}, we first describe a naive version of the Stochastic Approximation Cut Algorithm. The overall naive algorithm (Algorithm \ref{Al1}) is divided into two chains that are run in parallel. 

The auxiliary chain $h_n=(\tilde{\theta}_n,\tilde{\varphi}_n)$, $n=0,1,2,...,$ uses Stochastic Approximation Monte Carlo \citep{doi:10.1198/016214506000001202} to estimate $p(Y|\varphi)$ at a set of $m$ pre-selected auxiliary parameter values $\Phi_0=\{\varphi_0^{(1)},...,\varphi_0^{(m)}\}$. As we detail in the supplementary materials, the set $\Phi_0$ is chosen from a set of MCMC samples drawn from $p(\varphi|Z)$, chosen using the Max-Min procedure \citep{doi:10.1080/01621459.2015.1009072} that repeatedly adds the sample that has the largest Euclidean distance to the hitherto selected $\Phi_0$. This ensures that $\Phi_0$ covers the major part of the support of $p(\varphi|Z)$. A reasonably large $m$ ensures the distribution $p(\theta|Y,\varphi)$ overlaps each other for neighbouring $\varphi_0^{(a)}$ and $\varphi_0^{(b)}$. The target density for $(\tilde{\theta},\tilde{\varphi}) \in \Theta\times\Phi_0$, which is proportional to $p(\theta | Y, \varphi)$ in \eqref{E3} at the values $\Phi_0$, is
\begin{equation}
p(\tilde{\theta},\tilde{\varphi})=\frac{1}{m}\sum_{i=1}^m\frac{p(Y|\tilde{\theta},\varphi_0^{(i)})p(\tilde{\theta})}{p(Y|\varphi_0^{(i)})} \mathbbm{1}_{\{\tilde{\varphi}=\varphi_0^{(i)}\}}.
\label{E9}
\end{equation}
Given proposal distributions $q_1(\tilde{\theta}^\prime|\tilde{\theta})$ (e.g., symmetric random walk proposal) and $q_2(\tilde{\varphi}^\prime|\tilde{\varphi})$ (e.g., uniformly drawing $\tilde{\varphi}^\prime$ from a neighbouring set of $\tilde{\varphi}$) for $\tilde{\theta}$ and $\tilde{\varphi}$ individually, at each iteration $n$, proposals $\tilde{\theta}^\prime$ and $\tilde{\varphi}^\prime$ are drawn from a mixture proposal distribution, with a fixed mixing probability $p_{mix}$,
\[
q(\tilde{\theta}^\prime,\tilde{\varphi}^\prime|\tilde{\theta}_{n-1},\tilde{\varphi}_{n-1})=\left\{
             \begin{array}{lr}
             p_{mix} q_1(\tilde{\theta}^\prime|\tilde{\theta}_{n-1}),\ \text{for}\ \tilde{\theta}^\prime\neq\tilde{\theta}_{n-1} &  \\
             (1-p_{mix}) q_2(\tilde{\varphi}^\prime|\tilde{\varphi}_{n-1}),\ \text{for}\ \tilde{\varphi}^\prime\neq\tilde{\varphi}_{n-1} \\
             0,\ \text{otherwise}  
             \end{array}
\right.
\]
and accepted according to the Metropolis-Hastings acceptance probability with an iteration-specific target
\[
p_{n}(\tilde{\theta},\tilde{\varphi})\propto\sum_{i=1}^m \frac{p(Y|\tilde{\theta},\varphi_0^{(i)})p(\tilde{\theta})}{\tilde{w}_{n-1}^{(i)}}\mathbbm{1}_{\{\tilde{\varphi}=\varphi_0^{(i)}\}},\ \tilde{\theta}\in\Theta,\ \tilde{\varphi}\in\Phi_0.
\]
Here $\tilde{w}_n^{(i)}$ is the estimate of $p(Y|\varphi_0^{(i)})$, $i=1,...,m$, up to a constant, and $\tilde{w}_n=(\tilde{w}_n^{(1)},...,\tilde{w}_n^{(m)})$ is a vector of these estimates at each of the pre-selected auxiliary parameter values $\Phi_0$. We set $\tilde{w}_0^{(i)}=1$, $i=1,..., m$ at the start. As described in \cite{doi:10.1198/016214506000001202} and \cite{doi:10.1080/01621459.2015.1009072}, the estimates are updated by
\begin{equation}
\log(\tilde{w}_n^{(i)})=\log(\tilde{w}_{n-1}^{(i)})+\xi_n(e_{n,i}-m^{-1}),\ i=1,...,m,
\label{E11}
\end{equation}
where $e_{n,i}=1$ if $\tilde{\varphi}_n=\varphi_0^{(i)}$ and $e_{n,i}=0$ otherwise, and $\xi_n=n_0/\max(n_0,n)$ decreases to 0 when $n$ goes to infinity (the shrink magnitude $n_0$ is a user-chosen fixed constant). Note that in this auxiliary chain, when the number of iterations is sufficiently large, we are drawing $(\theta,\varphi)$ from \eqref{E9}. Hence, by checking if the empirical sampling frequency of each $\varphi_0^{(i)}\in \Phi_0$ strongly deviates from $m^{-1}$, we can detect potential nonconvergence of the auxiliary chain.

In the main Markov chain $(\theta_n,\varphi_n)$, $n=1,2,...$ we draw $\varphi^\prime$ from a proposal distribution $q(\varphi^\prime|\varphi)$, and then draw $\theta^\prime$ according to a random measure
\begin{equation}
P_n^\ast(\theta\in\mathcal{B}|Y,\varphi^\prime)=\frac{\sum_{j=1}^{n}\sum_{i=1}^m \tilde{w}_{j-1}^{(i)}\frac{p(Y|\tilde{\theta}_j,\varphi^\prime)}{p(Y|\tilde{\theta}_j,\varphi_0^{(i)})}\mathbbm{1}_{\{\tilde{\theta}_j\in\mathcal{B},\varphi_0^{(i)}=\tilde{\varphi}_j\}}}{\sum_{j=1}^{n}\sum_{i=1}^m \tilde{w}_{j-1}^{(i)}\frac{p(Y|\tilde{\theta}_j,\varphi^\prime)}{p(Y|\tilde{\theta}_j,\varphi_0^{(i)})}\mathbbm{1}_{\{\varphi_0^{(i)}=\tilde{\varphi}_j\}}},
\label{E12}
\end{equation}
where $\mathcal{B}\subset\Theta$ is any Borel set. Given a $\varphi$, the random measure \eqref{E12} is formed via a dynamic importance sampling procedure proposed in \cite{10.2307/3085723} with intention to approximate the unknown distribution $p(\theta|Y,\varphi)$ (see supplementary material for a detailed explanation). For any Borel set $\mathcal{B}\subset\Theta$, we have
\[
\begin{aligned}
	&\frac{1}{n}\sum_{j=1}^{n}\sum_{i=1}^m \tilde{w}_{j-1}^{(i)}\frac{p(Y|\tilde{\theta}_j,\varphi)p(\tilde{\theta}_j)}{p(Y|\tilde{\theta}_j,\varphi_0^{(i)})p(\tilde{\theta}_j)}\mathbbm{1}_{\{\tilde{\theta}_j\in\mathcal{B},\varphi_0^{(i)}=\tilde{\varphi}_j\}} \\
	&\rightarrow \sum_{i=1}^m \int_\mathcal{B} m p(Y|\varphi_0^{(i)}) \frac{p(Y|\theta,\varphi)p(\theta)}{p(Y|\theta,\varphi_0^{(i)})p(\theta)}\frac{1}{m}\frac{p(Y|\theta,\varphi_0^{(i)})p(\theta)}{p(Y|\varphi_0^{(i)})} d\theta \\
	&=m \int_\mathcal{B} p(Y|\theta,\varphi)p(\theta) d\theta,
\end{aligned}
\]
and similarly, the denominator of \eqref{E12} converges to the $m p(Y|\varphi)$. Hence, by Lemma 3.1 of \cite{doi:10.1080/01621459.2015.1009072}, since $\Theta\times\Phi$ is compact, for any Borel set $\mathcal{B}\subset\Theta$ and on any outcome $\omega$ of probability space $\Omega$, we have:
\begin{equation}
\lim_{n\rightarrow\infty}\sup_{\varphi\in\Phi}\left|P_n^\ast(\theta\in\mathcal{B}|Y,\varphi)-\int_\mathcal{B} p(\theta|Y,\varphi)d\theta\right|=0.
\label{E13}
\end{equation}
This implies that the distribution of $\{\theta_n\}$, drawn from \eqref{E12}, converges in distribution to $p(\theta | Y, \varphi)$, and this convergence occurs uniformly over $\Phi$. Note that, the probability measure $P_n^\ast(\theta\in\mathcal{B}|Y,\varphi^\prime)$ is adapted to filtration $\mathcal{G}_n=\sigma(\cup_{j=1}^{n} (\tilde{\theta}_j,\tilde{\varphi}_j,\tilde{w}_j))$ on $(\Omega,\mathcal{F},\mathbb{P})$, and has a Radon-Nikodym derivative with respect to a mixture of Dirac measures determined by $\tilde{\Theta}_n=\cup_{j=1}^{n}\{\tilde{\theta_j}\}$ \citep{doi:10.1198/106186008X386102}, because it is the law of a discrete random variable defined on $\tilde{\Theta}_n$. To achieve stronger convergence results, we will build a continuous probability distribution based on \eqref{E12}.

\begin{algorithm}[!htbp]
\SetAlgoLined
Initialize at starting points $h_0=(\tilde{\theta}_0,\tilde{\varphi}_0)$, $\tilde{w}_0$ and $(\theta_0,\varphi_0)$\;
 For $n=1,..., N$\;
 \begin{enumerate}[(a)]
 \item Auxiliary chain:
  \begin{enumerate}[(1)]
  \item Draw a proposal $(\tilde{\theta}^\prime,\tilde{\varphi}^\prime)$ according to $q(\tilde{\theta}^\prime,\tilde{\varphi}^\prime|\tilde{\theta}_{n-1},\tilde{\varphi}_{n-1})$.
  \item Accept the proposal, and set $(\tilde{\theta}_n,\tilde{\varphi}_n)=(\tilde{\theta}^\prime,\tilde{\varphi}^\prime)$ according to the iteration-specific acceptance probability.
  \item Calculate $\tilde{w}_n^{(i)}$ according to \eqref{E11}, $i=1,... ,m$.
  \end{enumerate}
 \item Main chain:
  \begin{enumerate}[(1)]
   \item Draw a proposal $\varphi^\prime$ according to $q(\varphi^\prime|\varphi_{n-1})$.
   \item Set $\varphi_n=\varphi^\prime$ with probability:
\[
\begin{aligned}
\alpha(\varphi^\prime|\varphi_{n-1})&= \min\left\lbrace 1, \frac{p(\theta^\prime|Y,\varphi^\prime)p(\varphi^\prime|Z)q(\varphi_{n-1}|\varphi^\prime)p(\theta_{n-1}|Y,\varphi_{n-1})}{p(\theta_{n-1}|Y,\varphi_{n-1})p(\varphi_{n-1}|Z)q(\varphi^\prime|\varphi_{n-1})p(\theta^\prime|Y,\varphi^\prime)}\right\rbrace \\
&=\min\left\lbrace 1, \frac{p(\varphi^\prime|Z)q(\varphi_{n-1}|\varphi^\prime)}{p(\varphi_{n-1}|Z)q(\varphi^\prime|\varphi_{n-1})}\right\rbrace.
\end{aligned}
\]
   \item If $\varphi^\prime$ is accepted, draw $\theta^\prime$ according to $P_n^\ast(\theta^\prime|Y,\varphi^\prime)$ defined in \eqref{E12} and set $\theta_n=\theta^\prime$.
   \item Otherwise if $\varphi^\prime$ is rejected, set $(\theta_n,\varphi_n)=(\theta_{n-1},\varphi_{n-1})$.
\end{enumerate}
 \end{enumerate}
 End For\;
 \caption{Naive Stochastic Approximation Cut Algorithm}
  \label{Al1}
\end{algorithm}

\subsection{Simple function approximation cut distribution}
\label{S2.2}
The convergence in distribution \eqref{E13} presented in the Naive Stochastic Approximation Cut Algorithm is not sufficiently strong to infer a law of large numbers or ergodicity of the drawn samples. We will show that these properties can be satisfied by targeting an approximation of the density function $p(\theta|Y,\varphi)$.

We adopt a density function approximation technique which uses a simple function as the basis. The use of a simple function to approximate a density function has been discussed previously \citep{Fu2002, doi:10.1080/03610910802657904}, but here we use a different partition of the support of the function, determined by rounding to a user-specified number of decimal places. The general theory is presented in the supplementary material. 

Given the $d$ dimensional compact set $\Theta$ and user-specified number of decimal places $\kappa$, we partition $\Theta$ in terms of (partial) hypercubes $\Theta_r$ whose centres $\theta_r$ are the rounded elements of $\Theta$ to $\kappa$ decimal places,
\begin{equation}
\Theta_r=\Theta\cap\{\theta:\left\|\theta-\theta_r\right\|_\infty\leq 5\times10^{-\kappa-1}\},\ r=1,...,R_\kappa,
\label{E15}
\end{equation}
where $R_\kappa$ is the total number of rounded elements. The boundary set $\bar{\Theta}_\kappa$, which has Lebesgue measure 0, is:
\[
\bar{\Theta}_\kappa = \Theta\cap\left(\bigcup_{r=1}^{R_\kappa} \{\theta:\left\|\theta-\theta_r\right\|_\infty= 5\times10^{-\kappa-1}\}\right).
\]
Using this partition, we are able to build a simple function density that approximates $p(\theta|Y,\varphi)$:
\[
p^{(\kappa)}(\theta|Y,\varphi) = \sum_{r=1}^{R_\kappa}\frac{1}{\mu(\Theta_r)}\int_{\Theta_r}p(\theta^\prime|Y,\varphi)d\theta^\prime \mathbbm{1}_{\{\theta\in\Theta_r\}},
\]
and let $P^{(\kappa)}$ be the corresponding probability measure on $\Theta$. The \textbf{simple function approximation cut distribution} is then formed by replacing the exact conditional distribution with this approximation
\[
p_{cut}^{(\kappa)}(\theta,\varphi)=p^{(\kappa)}(\theta|Y,\varphi)p(\varphi|Z).
\]
Let $P_{cut}^{(\kappa)}$ be the corresponding probability measure on $\Theta \times\Phi$. 

Given the general theory presented in the supplementary material, we have
\begin{equation}
    p^{(\kappa)}(\theta|Y,\varphi)\xrightarrow{\text{a.s.}}p(\theta|Y,\varphi),\ \ \ \text{as}\ \kappa\rightarrow\infty.
\label{E18}
\end{equation}
The rate of convergence is tractable if we further assume density $p(\theta|Y,\varphi)$ is continuously differentiable.

\begin{corollary}
\label{coro1}
If density function $p(\theta|Y,\varphi)$ is continuously differentiable, there exists a set $\mathscr{E}\subset\Theta$ with $\mu(\mathscr{E})=\mu(\Theta)$ such that the local convergence holds:
\[
|p^{(\kappa)}(\theta|Y,\varphi)-p(\theta|Y,\varphi)|\leq (\varepsilon(\theta,\kappa)+\left\|\nabla p(\theta|Y,\varphi)\right\|_2)\frac{\sqrt{d}}{10^\kappa},\ \forall \theta\in\mathscr{E},
\] 
where $\varepsilon(\theta,\kappa)\rightarrow 0$ as $\kappa\rightarrow\infty$.

In addition, the global convergence holds:
\[
\sup_{\theta\in\mathscr{E}}|p^{(\kappa)}(\theta|Y,\varphi)-p(\theta|Y,\varphi)|\leq \sup_{\theta\in\Theta}\left\|\nabla p(\theta|Y,\varphi)\right\|_2 \frac{\sqrt{d}}{10^\kappa}.
\]
\end{corollary}
\begin{proof}
See the general theory in the supplementary material.
\end{proof}

\subsection{Stochastic Approximation Cut Algorithm}
\label{S2.3}
We now refine the naive Stochastic Approximation Cut Algorithm by replacing in the main chain the proposal distribution $P_n^\ast$, which concentrates on the discrete set $\tilde{\Theta}_n$, by a distribution, with support on the compact set $\Theta$, that we will show converges almost surely to $P^{(\kappa)}$. 

Let $\mathscr{W}_n(\varphi)=(W_n(\Theta_1|Y,\varphi),...,W_n(\Theta_{R_\kappa}|Y,\varphi))$ be a random weight process based on the probability of the original proposal distribution $P_n^\ast$ taking a value in each partition component $\Theta_r$ as
\begin{equation}
W_n(\Theta_r|Y,\varphi)=\frac{P_n^\ast(\theta\in\Theta_r|Y,\varphi)+(n R_\kappa)^{-1}}{1+n^{-1}},
\label{E19}
\end{equation}
where $r=1,...,R_\kappa$. Note that $W_n(\Theta_r|Y,\varphi)$ is adapted to the auxiliary filtration $\mathcal{G}_n$. By adding a $(n R_\kappa)^{-1}$, each $W_n(\Theta_r|Y,\varphi)$, $r=1,...,R_\kappa$, is strictly positive and yet this modification does not affect the limit since $(n R_\kappa)^{-1}\rightarrow 0$. That is, on any outcome $\omega$ of probability space $\Omega$, we have
\begin{equation}
\lim_{n\rightarrow\infty}\sup_{\varphi\in\Phi;1\leq r\leq R_\kappa} \left|W_n(\Theta_r|Y,\varphi)-\int_{\Theta_r}p(\theta|Y,\varphi)d\theta\right|=0.
\label{E20}
\end{equation}

We now define the random measure process $P_n^{(\kappa)}$ that replaces $P_n^\ast$ used in the naive Stochastic Approximation Cut Algorithm. For any Borel set $\mathcal{B}$,
\begin{equation}
P_n^{(\kappa)}(\theta\in\mathcal{B}|Y,\varphi)=\int_{\mathcal{B}}\sum_{r=1}^{R_\kappa} \frac{1}{\mu(\Theta_r)} W_n(\Theta_r|Y,\varphi)\mathbbm{1}_{\{\theta\in\Theta_r\}}d\theta.
\label{E21}
\end{equation}
Clearly $P_n^{(\kappa)}(\theta\in\Theta|Y,\varphi)=1$ so $P_n^{(\kappa)}$ is a valid probability measure on $\Theta$. Additionally, since $\mathscr{W}_n(\varphi)$ is adapted to filtration $\mathcal{G}_n$, $P_n^{(\kappa)}$ is adapted to filtration $\mathcal{G}_n$. The Radon-Nikodym derivative of $P_n^{(\kappa)}$ with respect to the Lebesgue measure $\mu$ on $\Theta$ is
\begin{equation}
p_n^{(\kappa)}(\theta|Y,\varphi)=\sum_{r=1}^{R_\kappa} \frac{1}{\mu(\Theta_r)} W_n(\Theta_r|Y,\varphi)\mathbbm{1}_{\{\theta\in\Theta_r\}}.
\label{E22}
\end{equation}
This density is not continuous, but it is bounded on $\Theta $. In addition, since $\Theta$ is the support of $p(\theta|Y,\varphi)$ and $\mathscr{W}_n(\varphi)$ is strictly positive, the support of $p_n^{(\kappa)}$ is $\Theta $ for all $\varphi\in\Phi$ as well. 

Using $P_n^{(\kappa)}$ as the proposal distribution has the advantage that $p_n^{(\kappa)}$ converges almost surely to $p^{(\kappa)}$, in contrast to the convergence in distribution for the naive algorithm in \eqref{E13}.
\begin{lemma}
\label{lem1}
Given Assumption \ref{assu1}, on any outcome $\omega$ of probability space $\Omega$, we have:
\[
p_n^{(\kappa)}(\theta|Y,\varphi)\xrightarrow{\text{a.s.}}p^{(\kappa)}(\theta|Y,\varphi),
\]
and this convergence is uniform over $(\Theta\setminus\bar{\Theta}_\kappa) \times\Phi$.
\end{lemma}

Note that the convergence is to $p^{(\kappa)}(\theta|Y,\varphi)$ rather than $p(\theta|Y,\varphi)$, but we will show in Corollary \ref{coro2} that this bias reduces geometrically as the precision parameter $\kappa$ increases.

The complete Stochastic Approximation Cut Algorithm (SACut) is shown in Algorithm \ref{Al2}. The key idea is that we propose samples for $\theta$ from a density $p_n^{(\kappa)}(\theta|Y,\varphi)$, which approximates $p(\theta|Y,\varphi)$ and from which we can draw samples, but we accept these proposals according to $p^{(\kappa)}(\theta|Y,\varphi)$, which then cancels. This results in the acceptance probability being determined only by the proposal distribution for $\varphi$; the proposal distribution for $\theta$ is not involved. Indeed, the acceptance probability is the same as the partial Gibbs sampler that we will discuss in Section \ref{S3.1.1}. 

\begin{algorithm}[!htbp]
\SetAlgoLined
Initialize at starting points $h_0=(\tilde{\theta}_0,\tilde{\varphi}_0)$, $\tilde{w}_0$ and $(\theta_0,\varphi_0)$\;
 For $n=1,..., N$\;
 \begin{enumerate}[(a)]
 \item Auxiliary chain:
  \begin{enumerate}[(1)]
  \item Draw a proposal $(\tilde{\theta}^\prime,\tilde{\varphi}^\prime)$ according to $q(\tilde{\theta}^\prime,\tilde{\varphi}^\prime|\tilde{\theta}_{n-1},\tilde{\varphi}_{n-1})$.
  \item Accept the proposal, and set $(\tilde{\theta}_n,\tilde{\varphi}_n)=(\tilde{\theta}^\prime,\tilde{\varphi}^\prime)$ according to the iteration-specific acceptance probability.
  \item Calculate $\tilde{w}_n^{(i)}$ according to \eqref{E11}, $i=1,... ,m$.
  \end{enumerate}
 \item Main chain:
  \begin{enumerate}[(1)]
   \item Draw a proposal $\varphi^\prime$ according to $q(\varphi^\prime|\varphi_{n-1})$.
   \item Set $\varphi_n=\varphi^\prime$ with probability:
\[
\begin{aligned}
\alpha(\varphi^\prime|\varphi_{n-1})&= \min\left\lbrace 1, \frac{p^{(\kappa)}(\theta^\prime|Y,\varphi^\prime)p(\varphi^\prime|Z)q(\varphi_{n-1}|\varphi^\prime)p^{(\kappa)}(\theta_{n-1}|Y,\varphi_{n-1})}{p^{(\kappa)}(\theta_{n-1}|Y,\varphi_{n-1})p(\varphi_{n-1}|Z)q(\varphi^\prime|\varphi_{n-1})p^{(\kappa)}(\theta^\prime|Y,\varphi^\prime)}\right\rbrace \\
&=\min\left\lbrace 1, \frac{p(\varphi^\prime|Z)q(\varphi_{n-1}|\varphi^\prime)}{p(\varphi_{n-1}|Z)q(\varphi^\prime|\varphi_{n-1})}\right\rbrace.
\end{aligned}
\]
   \item If $\varphi^\prime$ is accepted, calculate $W_n(\Theta_r|Y,\varphi^\prime)$ defined in \eqref{E19}, $r=1,..., R_\kappa$. Draw a proposal $\theta^\prime$ according to $p_n^{(\kappa)}(\theta^\prime|Y,\varphi^\prime)$ defined in \eqref{E21} and set $\theta_n=\theta^\prime$.
   \item Otherwise if $\varphi^\prime$ is rejected, set  $(\theta_n,\varphi_n)=(\theta_{n-1},\varphi_{n-1})$.
\end{enumerate}
 \end{enumerate}
 End For\;
 \caption{Stochastic Approximation Cut Algorithm (SACut)}
 \label{Al2}
\end{algorithm}

\subsection{Parallelization and Simplification of Computation}
The  main computational bottleneck of the Stochastic Approximation Cut Algorithm is the updating and storage of the cumulative set of auxiliary variable values $\tilde{\Theta}_n=\cup_{j=1}^n\{\tilde{\theta}_j\}$. Since we draw a new $\varphi^\prime$ at each iteration, in order to calculate all possible probabilities defined by \eqref{E12} and \eqref{E19}, the density $p(Y|\tilde{\theta},\varphi^\prime)$ must be calculated $|\tilde{\Theta}_n|$ times. This is equivalent to running $|\tilde{\Theta}_n|$ internal iterations at each step of external iteration for the existing approximate approaches proposed in \cite{Plummer2015}. Note that $\tilde{\Theta}_n$ is solely generated from the auxiliary chain so $|\tilde{\Theta}_n|$ is not affected by the precision parameter $\kappa$. If the calculation of this density is computationally expensive, the time to perform each update of the chain will become prohibitive when $|\tilde{\Theta}_n|$ is large. However, the calculation of $p(Y|\tilde{\theta},\varphi^\prime)$ for different values of $\tilde{\theta}$ is embarrassingly parallel so can be evaluated in parallel whenever multiple computer cores are available, enabling a considerable speed up.

The speed of the computation can be further improved by reducing the size of $\tilde{\Theta}_n$. Given the precision parameter $\kappa$, we round all elements from $\tilde{\Theta}_n$ to their $\kappa$ decimal places and let $\tilde{\Theta}_n^{(\kappa)}$ be the set that contains these rounded elements. At each iteration, the number of calculations of density $p(Y|\tilde{\theta},\varphi^\prime)$ is equal to the number of $d$-orthotopes that auxiliary chain $\{\tilde{\theta}_j\}_{j=1}^n$ has visited up to iteration $
n$ and by \eqref{E13} we know that the distribution of auxiliary samples of $\tilde{\theta}$ converges to the true distribution $p(\theta|Y,\varphi)$. Hence, the computational speed is mainly determined by the precision parameter $\kappa$ and the target distribution $p(\theta|Y,\varphi)$. In particular, for any fixed $\kappa$ and a sufficiently long auxiliary chain the computational cost is upper bounded by the case of uniform distribution since it equally distributes over the space $\Theta$.
\begin{theorem}
Given an arbitrary $d$ dimensional compact parameter space $\Theta$ and a precision parameter $\kappa$ and suppose that the auxiliary chain has converged before we start collecting auxiliary variable $\tilde{\theta}$, for any fixed number of iteration $n$. Then the expected number of $d$-orthotopes visited $\mathbb{E}\left(|\tilde{\Theta}_n^{(\kappa)}|\right)$ is maximized when the target distribution is uniform distribution.
\end{theorem}
\begin{proof}
See supplementary material (Online Resource 1).
\end{proof}
For example, given a $d$ dimensional parameter space $\Theta=[0-5\times10^{-\kappa-1},1+5\times10^{-\kappa-1}]^d$ and its partition $\Theta_r$, $r=1,...,11^{d\kappa}$, we consider the uniform distribution as the target distribution. Assume the auxiliary chain has converged, the expectation of $|\tilde{\Theta}_n^{(\kappa)}|$ is
\[
\mathbb{E}\left(|\tilde{\Theta}_n^{(\kappa)}|\right) = 11^{d\kappa}-\frac{\left(11^{d\kappa}-1\right)^n}{11^{d\kappa(n-1)}}.
\]
In the case of $d=1$, Figure \ref{F4} compares the number of orthotopes visited between the uniform distribution and truncated normal distribution when the standard deviation is 0.1 and 0.05. It shows that larger precision parameter $\kappa$ means more evaluations of $p(Y|\tilde{\theta},\varphi^\prime)$ are required. Hence, a wise choice of a small $\kappa$ can significantly reduce computation time. 

While small $\kappa$ means a loss of precision since local variations of original target distribution are smoothed by rounding the value of its samples, in most applied settings only a small number of significant figures are meaningful, and so the ability to trade-off the precision and computational speed is appealing. Comparing  short preliminary run of chains for different candidates of $\kappa$ may be useful when a suitable choice of $\kappa$ is unclear. We will discuss this in Section \ref{Sec4_1}.

\begin{figure}[tp] 
\centering 
\includegraphics[width=\textwidth]{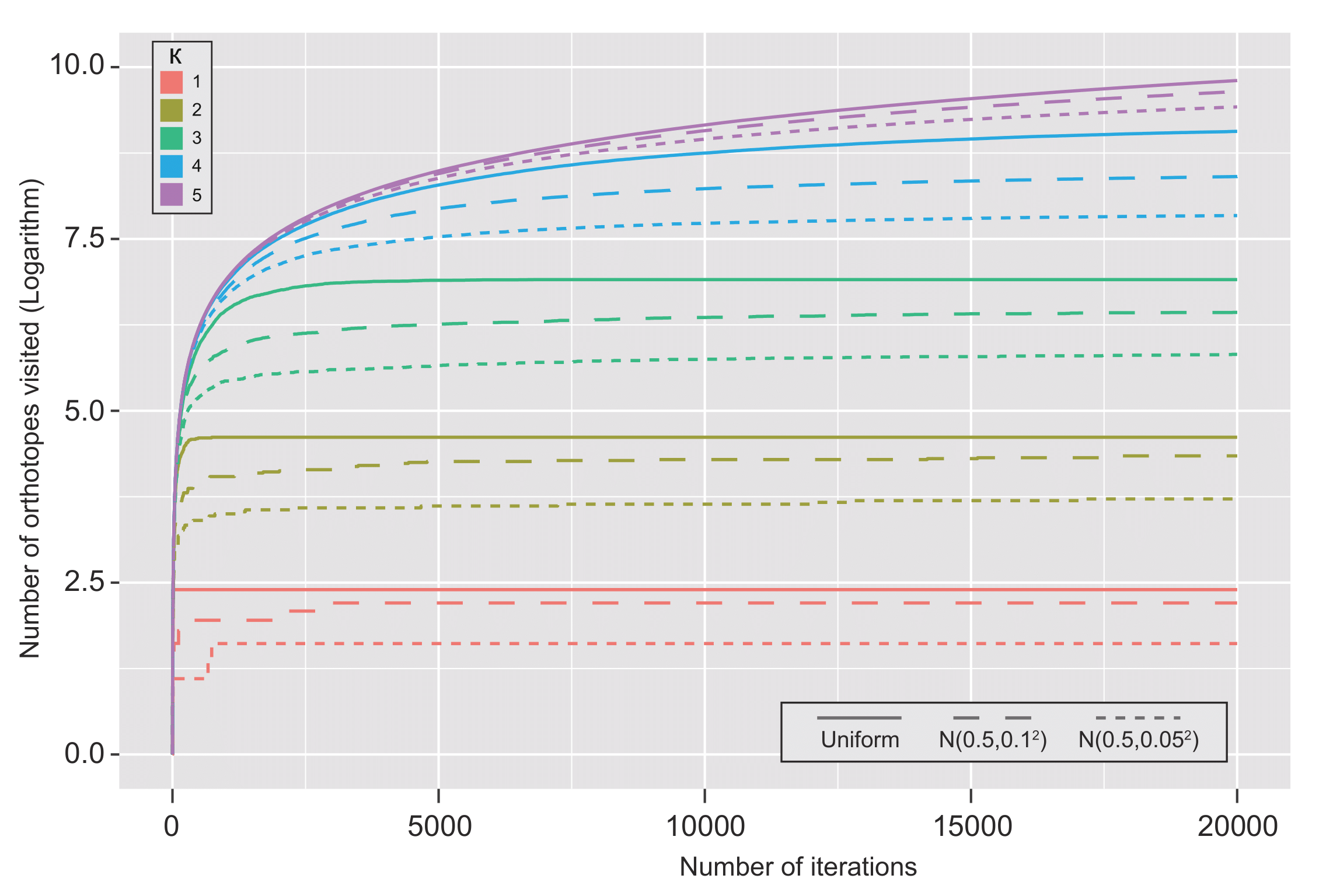}
\caption[Relationship between the number of orthotopes visited and the number of iterations]{Relationship between the number of orthotopes visited and the number of iterations when precision parameter $\kappa=1,2,3,4,5$. Separate Monte Carlo simulations were conducted for uniform distribution and truncated normal distribution with standard deviation 0.1 and 0.05. }
\label{F4}
\end{figure}

\section{Convergence Properties}
\label{sec:convergence}
In this section, we study the convergence properties of samples drawn by the Stochastic Approximation Cut Algorithm. We establish a weak law of large numbers with respect to the simple function approximation cut distribution $P_{cut}^{(\kappa)}$, under some regularity conditions, by proving that the conditions required by Theorem 3.2 in \cite{doi:10.1080/01621459.2015.1009072} are satisfied. We then prove that the bias with respect to $P_{cut}$ can be reduced geometrically by increasing the precision parameter $\kappa$. To aid exposition of the convergence properties, it is necessary to first introduce two simpler but infeasible alternative algorithms. Then we prove the convergence of the algorithm.

The framework of our proofs follow \cite{doi:10.1080/01621459.2015.1009072}. However, adjustments are made for two key differences. Firstly, the parameter of interest here has two components, instead of just one, and we require completely different proposal distributions to those in \cite{doi:10.1080/01621459.2015.1009072}: the proposal distribution of $\theta$ involves an auxiliary chain and simple function approximation, and the proposal distribution of $\varphi$ is a standard MCMC algorithm. Secondly, the parameter drawn by \eqref{E22} here is retained, rather than being discarded as in \cite{doi:10.1080/01621459.2015.1009072}. This means the distributions involved here are different and more complicated.

\subsection{Infeasible Alternative Algorithms}
\begin{definition}
Given a signed measure $\mathcal{M}$ defined on a set $E$, and a Borel set $\mathcal{B}\subset E$, define the total variation norm of $\mathcal{M}$  as
\[
\left\|\mathcal{M}(\cdot)\right\|_{TV}=\sup_{\mathcal{B}\subset E} \left|\mathcal{M}(\mathcal{B})\right|.
\]
\end{definition}

\subsubsection{A Partial Gibbs Sampler}
\label{S3.1.1}

The most straightforward algorithm that draws samples from $p_{cut}^{(\kappa)}(\theta,\varphi)$ is a standard partial Gibbs sampler, which draws proposals $\theta^\prime$ from $p^{(\kappa)}(\theta^\prime|Y,\varphi^\prime)$, given a $\varphi^\prime$ drawn from a proposal distribution $q(\varphi^\prime|\varphi_{n-1})$. The transition kernel is
\[
\resizebox{\hsize}{!}{$\begin{aligned}
&\textbf{u}^{(1)}((\theta_n,\varphi_n)|(\theta_{n-1},\varphi_{n-1})) \\
&=\alpha(\varphi_n|\varphi_{n-1})p^{(\kappa)}(\theta_n|Y,\varphi_n)q(\varphi_n|\varphi_{n-1}) \\
&\ +\left(1-\int_{\Theta \times\Phi}\alpha(\varphi|\varphi_{n-1})p^{(\kappa)}(\theta|Y,\varphi)q(\varphi|\varphi_{n-1})d\theta d\varphi\right)\delta\left((\theta_n,\varphi_n)-(\theta_{n-1},\varphi_{n-1})\right) \\
&=\alpha(\varphi_n|\varphi_{n-1})p^{(\kappa)}(\theta_n|Y,\varphi_n)q(\varphi_n|\varphi_{n-1}) \\
&\ +\left(1-\int_\Phi \alpha(\varphi|\varphi_{n-1})q(\varphi|\varphi_{n-1})d\varphi\right)\delta\left((\theta_n,\varphi_n)-(\theta_{n-1},\varphi_{n-1})\right),
\end{aligned}$}
\]
where $\delta$ is the multivariate Dirac delta function and
\[
\alpha(\varphi_n|\varphi_{n-1})=\min\left\lbrace 1, \frac{p(\varphi_n|Z)q(\varphi_{n-1}|\varphi_n)}{p(\varphi_{n-1}|Z)q(\varphi_n|\varphi_{n-1})}\right\rbrace.
\]
This transition kernel is Markovian and admits $p_{cut}^{(\kappa)}$ as its stationary distribution, provided a proper proposal distribution $q(\varphi_n|\varphi_{n-1})$ is used. We write $\textbf{U}^{(1)}$ for the corresponding probability measure.

Let $\textbf{u}^{(s)}$ denote the s-step transition kernel and write $\textbf{U}^{(s)}$ for the corresponding probability measure. By \cite{meyn_tweedie_glynn_2009}, we have ergodicity on $\Theta \times\Phi$,
\[
\lim_{s\rightarrow\infty} \left\|\textbf{U}^{(s)}(\cdot)-P_{cut}^{(\kappa)}(\cdot)\right\|_{TV}=0,
\]
and for any bounded function $f$ defined on $\Theta \times\Phi$, we have a strong law of large numbers
\[
\frac{1}{N}\sum_{n=1}^N f(\theta_n,\varphi_n)\xrightarrow{\text{a.s.}} \int_{\Theta \times\Phi} f(\theta,\varphi)P_{cut}^{(\kappa)}(d\theta,d\varphi).
\]

Note, however, that this algorithm is infeasible because $p^{(\kappa)}(\theta|Y,\varphi)$ is intractable, since $p(\theta|Y,\varphi)$ is intractable, and so we cannot directly draw proposals for $\theta$.

\subsubsection{An Adaptive Metropolis-Hastings Sampler}

An adaptive Metropolis-Hastings sampler can be built by replacing $p^{(\kappa)}$ in the calculation of acceptance probability of the Stochastic Approximation Cut Algorithm by its approximation $p_n^{(\kappa)}$, which is the exact proposal distribution for $\theta$ at the $n^{th}$ step. The acceptance probability is determined by both $\theta$ and $\varphi$,
\[
\alpha_n((\theta^\prime,\varphi^\prime)|(\theta_{n-1},\varphi_{n-1}))=\min\left\lbrace 1, \frac{p^{(\kappa)}(\theta^\prime|Y,\varphi^\prime)p(\varphi^\prime|Z)q(\varphi_{n-1}|\varphi^\prime)p_n^{(\kappa)}(\theta_{n-1}|Y,\varphi_{n-1})}{p^{(\kappa)}(\theta_{n-1}|Y,\varphi_{n-1})p(\varphi_{n-1}|Z)q(\varphi^\prime|\varphi_{n-1})p_n^{(\kappa)}(\theta^\prime|Y,\varphi^\prime)}\right\rbrace.
\]
and we can write the transition kernel,
\[
\begin{aligned}
&\textbf{v}_n^{(1)}((\theta_n,\varphi_n)|(\theta_{n-1},\varphi_{n-1}),\mathcal{G}_n) =\alpha_n((\theta_n,\varphi_n)|(\theta_{n-1},\varphi_{n-1}))p_n^{(\kappa)}(\theta_n|Y,\varphi_n)q(\varphi_n|\varphi_{n-1}) \\
&\ +\left(1-\int_{\Theta \times\Phi}\alpha_n((\theta,\varphi)|(\theta_{n-1},\varphi_{n-1}))p_n^{(\kappa)}(\theta|Y,\varphi)q(\varphi|\varphi_{n-1})d\theta d\varphi\right)\delta\left((\theta_n,\varphi_n)-(\theta_{n-1},\varphi_{n-1})\right),
\end{aligned}
\] 
where $\delta$ is the multivariate Dirac delta function. Conditional on the filtration $\mathcal{G}_n$, $\textbf{v}_n^{(1)}$ is Markovian. We write $\textbf{V}_n^{(1)}$ for the corresponding probability measure. Note that this sampler is not a standard Metropolis-Hastings algorithm since the transition kernel is not constant. Instead, it is an \emph{external} adaptive MCMC algorithm \citep{Adaptive_2011}.

Given information up to $\mathcal{G}_n$, if we stop updating auxiliary process then $P_n^{(\kappa)}$ is fixed and not random, and this sampler reduces to a standard Metropolis-Hastings sampler. The transition kernel $\textbf{V}_n^{(1)}$ admits $p_{cut}^{(\kappa)}$ as its stationary distribution provided a proper proposal distribution is used. That is, define
\[
\textbf{v}_n^{(s)}=\int_{\Theta^{s-1}\times\Phi^{s-1}}\prod_{k=1}^s \textbf{v}_n^{(1)}((\theta_k,\varphi_k)|(\theta_{k-1},\varphi_{k-1}),\mathcal{G}_n)d\theta_{1:s-1}d\varphi_{1:s-1},
\]
and $\textbf{V}_n^{(s)}$ as the corresponding probability measure. Then on $\Theta \times\Phi$ we have
\[
\lim_{s\rightarrow\infty}\left\|\textbf{V}_n^{(s)}(\cdot)-P_{cut}^{(\kappa)}(\cdot)\right\|_{TV}=0.
\]

Note, however, that this algorithm is also infeasible because, while we can draw proposals for $\theta$, since $p_n^{(\kappa)}$ is known up to $\mathcal{G}_n$, $p^{(\kappa)}(\theta|Y,\varphi)$ remains intractable so we cannot calculate the acceptance probability.

\subsection{Convergence of the Stochastic Approximation Cut Algorithm}
The infeasibility of the partial Gibbs sampler and the adaptive Metropolis-Hastings sampler motivate the development of the Stochastic Approximation Cut Algorithm, which replaces the proposal distribution $p_n^{(\kappa)}$ by its target $p^{(\kappa)}$ in the accept-reject step of the adaptive Metropolis-Hastings sampler. This leads to the same acceptance probability as is used by the partial Gibbs sampler so the proposed algorithm can be viewed as combining the advantages of both the partial Gibbs sampler and the adaptive Metropolis-Hastings sampler. The transition kernel of the Stochastic Approximation Cut Algorithm is
\[
\begin{aligned}
&\textbf{t}_n^{(1)}((\theta_n,\varphi_n)|(\theta_{n-1},\varphi_{n-1}),\mathcal{G}_n) \\
&=\alpha(\varphi_n|\varphi_{n-1})p_n^{(\kappa)}(\theta_n|Y,\varphi_n)q(\varphi_n|\varphi_{n-1}) \\
&\ +\left(1-\int_{\Theta \times\Phi}\alpha(\varphi|\varphi_{n-1})p_n^{(\kappa)}(\theta|Y,\varphi)q(\varphi|\varphi_{n-1})d\theta d\varphi\right)\delta\left((\theta_n,\varphi_n)-(\theta_{n-1},\varphi_{n-1})\right) \\
&=\alpha(\varphi_n|\varphi_{n-1})p_n^{(\kappa)}(\theta_n|Y,\varphi_n)q(\varphi_n|\varphi_{n-1}) \\
&\ +\left(1-\int_\Phi \alpha(\varphi|\varphi_{n-1})q(\varphi|\varphi_{n-1})d\varphi\right)\delta\left((\theta_n,\varphi_n)-(\theta_{n-1},\varphi_{n-1})\right),
\end{aligned}
\] 
where $\delta$ is the multivariate Dirac delta function. Conditionally to $\mathcal{G}_n$, the transition kernel $\textbf{t}_n^{(1)}$ is Markovian. We write $\textbf{T}_n^{(1)}$ for the corresponding probability measure. Given information up to $\mathcal{G}_n$ and stopping updating the auxiliary process, $P_n^{(\kappa)}$ is fixed and not random, and we define the $s$-step transition kernel as
\[
\textbf{t}_n^{(s)}=\int_{\Theta^{s-1}\times\Phi^{s-1}}\prod_{k=1}^s \textbf{t}_n^{(1)}((\theta_k,\varphi_k)|(\theta_{k-1},\varphi_{k-1}),\mathcal{G}_n)d\theta_{1:s-1}d\varphi_{1:s-1},
\]
and write $\textbf{T}_n^{(s)}$ for the corresponding probability measure.

We now present several lemmas required to prove a weak law of large numbers for this algorithm (proofs in supplementary material (Online Resource 1)), appropriately modifying the reasoning of \cite{10.2307/2245077}, \cite{10.2307/2337435} and \cite{doi:10.1080/01621459.2015.1009072} for this setting.
\begin{assumption}
\label{assu2}
The posterior density $p(\varphi|Z)$ is continuous on $\Phi$ and the proposal distribution $q(\varphi^\prime|\varphi)$ is continuous with respect to $(\varphi^\prime,\varphi)$ on $\Phi\times\Phi$.
\end{assumption}

\begin{lemma}[Diminishing Adaptation]
\label{lem2}
Given Assumptions \ref{assu1} and \ref{assu2}, then
\[
\lim_{n\rightarrow\infty}\sup_{\theta\in\Theta ,\varphi\in\Phi}\left\|\textbf{T}_{n+1}^{(1)}\left(\cdot|(\theta,\varphi),\mathcal{G}_{n+1}\right)-\textbf{T}_n^{(1)}\left(\cdot|(\theta,\varphi),\mathcal{G}_n\right)\right\|_{TV}=0.
\]
\end{lemma}

Before presenting the next lemma, we introduce the concept of \emph{local positivity}. 
\begin{definition}
A proposal distribution $q(\psi^\prime|\psi)$ satisfies local positivity if there exists $\delta>0$ and $\varepsilon>0$ such that for every $\psi\in\Psi$, $|\psi^\prime-\psi|\leq\delta$ implies that $q(\psi^\prime|\psi)>\varepsilon$.
\end{definition}

\begin{lemma}
\label{lem3}
Given Assumption \ref{assu1}, the proposal distributions with densities $p_n^{(\kappa)}:\Theta \rightarrow\mathbb{R}$ and $p^{(\kappa)}:\Theta \rightarrow\mathbb{R}$ are both uniformly lower bounded away from $0$ and satisfy local positivity uniformly for all values $\varphi\in\Phi$.
\end{lemma}

\begin{lemma}[Stationarity]
\label{lem4}
Given Assumptions \ref{assu1} and \ref{assu2}, and the filtration $\mathcal{G}_n$ (i.e. $P_n^{(\kappa)}$ is not random), then if the transition kernel measures $\textbf{U}^{(1)}$ and $\textbf{V}_n^{(1)}$ both admit an irreducible and aperiodic Markov chain, then the transition kernel measure $\textbf{T}_n^{(1)}$ admits an irreducible and aperiodic chain. Moreover, if the proposal distribution $q(\varphi^\prime|\varphi)$ satisfies local positivity, then there exists a probability measure $\Pi_n$ on $\Theta \times\Phi$ such that for any $(\theta_0,\varphi_0)\in\Theta \times\Phi$,
\[
\lim_{s\rightarrow\infty}\left\|\textbf{T}_n^{(s)}(\cdot)-\Pi_n\left(\cdot\right)\right\|_{TV}=0,
\]
and this convergence is uniform over $\Theta\times\Phi$.
\end{lemma} 

\begin{lemma}[Asymptotic Simultaneous Uniform Ergodicity]
\label{lem5}
Given Assumptions \ref{assu1} and \ref{assu2} and the assumptions in Lemma \ref{lem4}, for any initial value $(\theta_0,\varphi_0)\in\Theta\times\Phi$, and any $\varepsilon>0$ and $e>0$, there exist constants $S(\varepsilon)>0$ and $N(\varepsilon)>0$ such that
\[
\mathbb{P}\left(\left\lbrace P_n^{(\kappa)}:\ \left\|\textbf{T}_n^{(s)}\left(\cdot\right)-P_{cut}^{(\kappa)}\left(\cdot\right)\right\|_{TV}\leq\varepsilon\right\rbrace\right)>1-e,
\]
for all $s>S(\varepsilon)$ and $n>N(\varepsilon)$.
\end{lemma}

Lemma \ref{lem2} leads to condition (c) (Diminishing Adaptation), Lemma \ref{lem4} leads to condition (a) (Stationarity) and Lemma \ref{lem5} leads to condition (b) (Asymptotic Simultaneous Uniform Ergodicity) in Theorem 3.2 of \cite{doi:10.1080/01621459.2015.1009072}. Hence, we have the following weak law of large numbers.

\begin{theorem}[WLLN]
\label{them2}
Suppose that the conditions of Lemma \ref{lem5} hold. Let $f$ be any measurable bounded function on $\Theta\times\Phi$. Then for samples $(\theta_n,\varphi_n)$, $n=1,2,...$ drawn using the Stochastic Approximation Cut Algorithm, we have that
\[
\frac{1}{N}\sum_{n=1}^N f(\theta_n,\varphi_n)\rightarrow \int_{\Theta\times\Phi} f(\theta,\varphi)P_{cut}^{(\kappa)}(d\theta,d\varphi), \ \text{in probability}.
\]
\end{theorem}
\begin{proof}
This follows by Theorem 3.2 in \cite{doi:10.1080/01621459.2015.1009072}
\end{proof}

Given further conditions and combining Corollary \ref{coro1} with Theorem \ref{them2} we have the following corollary.
\begin{corollary}
\label{coro2}
Given the conditions in Corollary \ref{coro1} hold for the cut distribution $p_{cut}$ and conditions in Theorem \ref{them2} hold. Then given a measurable and bounded function $f: \Theta\times\Phi\rightarrow\mathbb{R}$, there exists, for any $\varepsilon>0$ and $e>0$, a precision parameter $\kappa$ and iteration number $N$, such that for samples $(\theta_n,\varphi_n)$, $n=1,2,...$ drawn using the Stochastic Approximation Cut Algorithm, we have that
\[
\mathbb{P}\left(\left|\frac{1}{N}\sum_{n=1}^N f(\theta_n,\varphi_n)-\int_{\Theta\times\Phi} f(\theta,\varphi)P_{cut}(d\theta,d\varphi)\right|\leq\varepsilon\right)>1-e.
\]
More specifically, the bias
\[
\left|\int_{\Theta\times\Phi}f(\theta,\varphi)P_{cut}(d\theta,d\varphi)-\int_{\Theta\times\Phi}f(\theta,\varphi)P_{cut}^{(\kappa)}(d\theta,d\varphi)\right|
\]
can be controlled by
\[
\sup_{\theta\in\Theta,\varphi\in\Phi}\left\|\nabla_{\theta}p(\theta|Y,\varphi)\right\|_2 \frac{\sqrt{d}}{10^\kappa}\left(\int_{\Theta\times\Phi} f(\theta,\varphi)p(\varphi|Z) d\theta d\varphi\right).
\]
\end{corollary}

Corollary \ref{coro2} shows that, although the convergence established by Theorem \ref{them2} is biased with respect to the true cut distribution $P_{cut}$, the bias can be geometrically reduced by selecting a large precision parameter $\kappa$.

\section{Illustrative Examples}
\label{sec:example}
We demonstrate the proposed algorithm in this section. First, we use a simulation example to introduce a simple method for choosing the precision parameter $\kappa$, and demonstrate that the proposed algorithm can eliminate the feedback from a suspect module. We then examine a simulated case designed to highlight when existing algorithms will perform poorly. We finally apply our algorithm to an epidemiological example and compare results with existing studies. The R package \emph{SACut} and code to replicate these examples can be downloaded from GitHub\footnotemark[1].
\footnotetext[1]{\url{https://github.com/MathBilibili/Stochastic-approximation-cut-algorithm}}

\begin{figure}[tp] 
\centering 
\includegraphics[width=\textwidth]{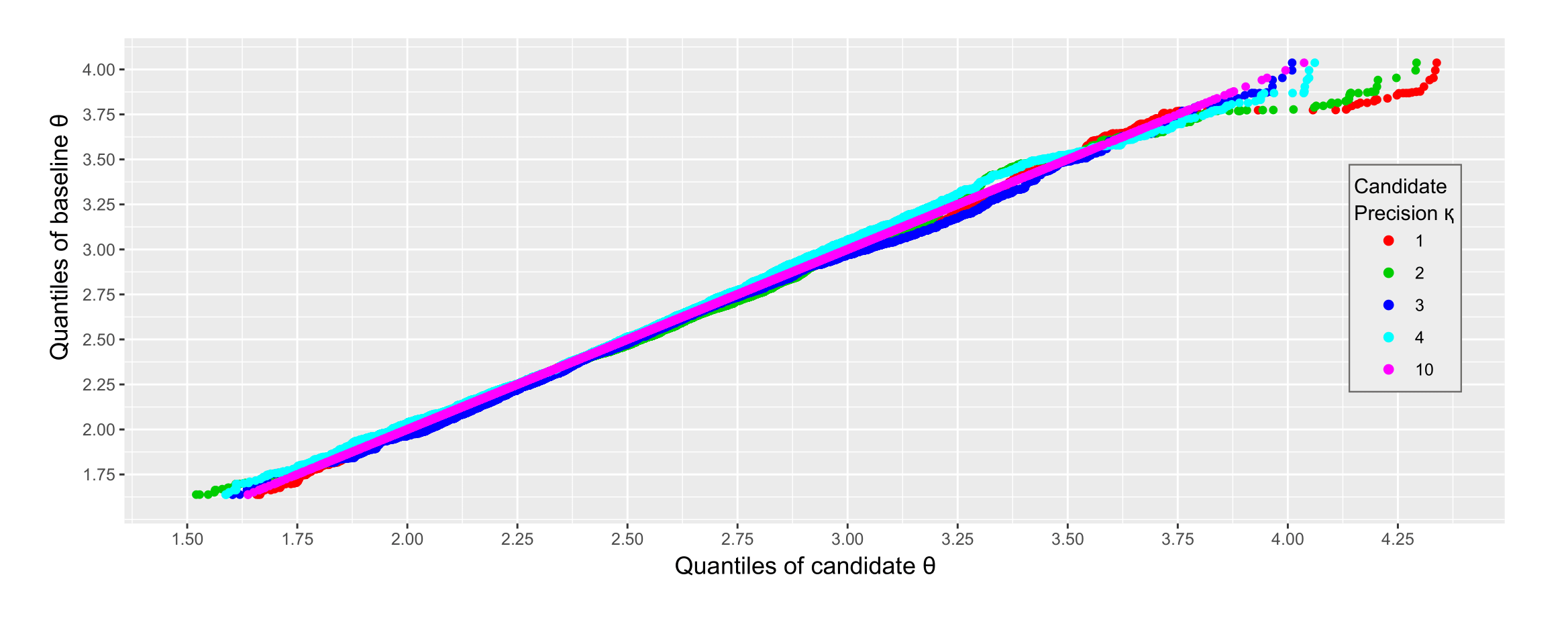}
\caption[Quantile-quantile plot for $\theta$ drawn from \eqref{E35} with precision parameter $\kappa=1,2,3,4,10$.]{Quantile-quantile plot for $\theta$ drawn from \eqref{E35} with precision parameter $\kappa=1,2,3,4,10$. The x-axis of the quantile-quantile plot is the quantile of samples under different $\kappa$, the y-axis is the quantile of samples under the gold standard $\kappa=10$.}
\label{F5}
\end{figure}

\subsection{Simulated Random Effects Example}
\label{Sec4_1}
In this example, we discuss a simple method for selecting the precision parameter $\kappa$ and show that the proposed algorithm can effectively cut the feedback from a suspect module. 

We consider a simple normal-normal random effect example previously discussed by \cite{liu2009modularization}, with groups $i=1,...,100=N$, observations $Y_{ij} \sim \text{N}(\beta_i,\varphi_i^2)$, $j=1,...,20$ in each group, and random effects distribution $\beta_i \sim \text{N}(0,\theta^2)$. Our aim is to estimate the random effects standard deviation $\theta$ and the residual standard deviation $\varphi=(\varphi_1,...,\varphi_N)$. By sufficiency, the likelihood can be equivalently represented in terms of the group-specific means $\bar{Y}_i=\frac{1}{20}\sum_{j=1}^{20} Y_{ij}$ and the sum of squared deviations $s_i^2=\sum_{j=1}^{20} (Y_{ij}-\bar{Y}_i)^2$ as
\[
\begin{aligned}
&\bar{Y}_i\sim \text{N}(\beta_i,\frac{\varphi_i^2}{20}), \\
&s_i^2\sim \text{Gamma}\left(\frac{20-1}{2},\frac{1}{2 \varphi_i^2}\right).
\end{aligned}
\]
Given the sufficient statistics $\bar{Y}=(\bar{Y}_1,...,\bar{Y}_N)$ and $s^2=(s_1^2,...,s_N^2)$, the model consists of two modules: module 1 involving $(s^2, \varphi)$ and module 2 involving $(\bar{Y}, \beta, \varphi)$, where $\beta=(\beta_1,...,\beta_N)$. 

We consider the situation when an outlier group is observed, meaning that module 2 is misspecified, and compare the standard Bayesian posterior distribution with the cut distribution. Specifically, we simulate data from the model with $\theta^2=2$, and $\varphi_i^2$ drawn from a $\text{Unif}(0.5, 1.5)$ distribution ($\varphi_1^2=1.60$), but we artificially set $\beta_1 = 10$, making the first group an outlier and thus our normal assumption for the random effects misspecified.  Given priors $p(\varphi_i^2)\propto (\varphi_i^2)^{-1}$ and $p(\theta^2|\varphi^2)\propto (\theta^2+\bar{\varphi}^2/20)^{-1}$, \cite{liu2009modularization} showed the standard Bayesian marginal posterior distribution for the parameters of interest is:
\[
\begin{aligned}
&p(\theta,\varphi|\bar{Y},s^2)=p(\theta|\bar{Y},\varphi)p(\varphi|\bar{Y},s^2) \\
&\propto\frac{1}{\theta^2+\bar{\varphi}^2/20}\prod_{i=1}^{100} (\varphi_i^2)^{-\frac{21}{2}}\exp\left(-\frac{s_i^2}{2\varphi_i^2}\right)\frac{1}{(\theta^2+\varphi_i^2/20)^{1/2}}\exp\left(-\frac{\bar{Y}_i^2}{2(\theta^2+\varphi_i^2/20)}\right).
\end{aligned}
\]

Since we are confident about our assumption of normality of $Y_{ij}$ but not confident about our distributional assumption for the random effects $\beta_i$, following \cite{liu2009modularization}, we consider the cut distribution in which we remove the influence of $\bar{Y}$ on $\varphi$, so that possible misspecification of the first module does not affect $\varphi$:
\[
p_{cut}(\theta,\varphi):=p(\theta|\bar{Y},\varphi)p(\varphi|s^2),
\]
where
\[
p(\varphi|s^2) \propto \prod_{i=1}^{100} \varphi_i^{-21}\exp(-\frac{s_i^2}{2\varphi_i^2}).
\]

To apply the proposed algorithm we first construct the auxiliary parameter set for the parameter $\varphi$ by selecting 70 samples selected from posterior samples of $p(\varphi|s^2)$ by the Max-Min procedure \citep{doi:10.1080/01621459.2015.1009072}. We set the shrink magnitude $n_0=1000$ and run only the auxiliary chain for $10^4$ iterations before starting to store the auxiliary variable $h_n$, as suggested by \cite{doi:10.1080/01621459.2015.1009072}.

The precision parameter $\kappa$ should be chosen large enough to obtain accurate results, whilst being small enough that computation is not prohibitively slow. To illustrate this, we compare results with $\kappa=10$, which we regard as the gold standard, to results with $\kappa = 1, 2, 3, 4$. Different values of $\kappa$ affect the sampling of $\theta$ only via \eqref{E21}, so we compare samples drawn from $p_n^{(\kappa)}(\theta|\bar{Y},\varphi)$, averaged over the marginal cut distribution of $\varphi$:
\begin{equation}
p_n^{(\kappa)}(\theta|\bar{Y},s^2):=\int p_n^{(\kappa)}(\theta|\bar{Y},\varphi)p_{cut}(\varphi) d\varphi,
\label{E35}
\end{equation}
where the marginal cut distribution $p_{cut}(\varphi)$ is
\[
p_{cut}(\varphi):=\int p_{cut}(\theta,\varphi)d\theta=p(\varphi|s^2)\propto p(s^2|\varphi)p(\varphi).
\]

We draw $10^5$ samples from \eqref{E35} for each value of $\kappa$, after running the proposed algorithm with few iterations ($10^4$) as a preliminary trial. Figure \ref{F5} shows the quantile-quantile plot for 5 choices for $\kappa$. The fit appears good for all choices of $\kappa$, except in the tails, where $\kappa=3$ and $\kappa=4$ provide a closer match to the gold standard. Thus, we choose $\kappa=3$ as it gives a sufficiently accurate approximation.

We apply both the standard Bayesian approach and the Stochastic Approximation Cut Algorithm ($\kappa=3$), each with 10 independent chains. All chains were run for $10^5$ iterations and we retain only every $100^{th}$ value, after discarding the first $10\%$ of the samples, and we summarise the results by the mean and credible interval (CrI). Pooling the 10 chains for the cut distribution gave estimates of $\theta^2 = 2.54$ ($95\%$ CrI 1.93 - 3.44) and $\varphi_1^2 = 1.58$ ($95\%$ CrI 0.88 - 3.18), whereas the standard Bayesian approach gave estimates of $\theta^2 = 2.53$ ($95\%$ CrI 1.93 - 3.44) and $\varphi_1^2 = 1.69$ ($95\%$ CrI 0.91 - 3.76). Figure \ref{F6} presents the medians for the parameter of interest $\varphi_1^2$ under each of the 10 independent runs for the cut distribution and the standard Bayesian posterior. Recalling the true value for $\varphi_1^2 = 1.60$, it is clear that when using the Stochastic Approximation Cut algorithm the medians locate around its true value rather than deviating systematically towards one side. This indicates the proposed algorithm has successfully prevented the outlying observation from influencing the estimation of $\varphi_1^2$.

\begin{figure}[t] 
\centering 
\includegraphics[scale=0.3]{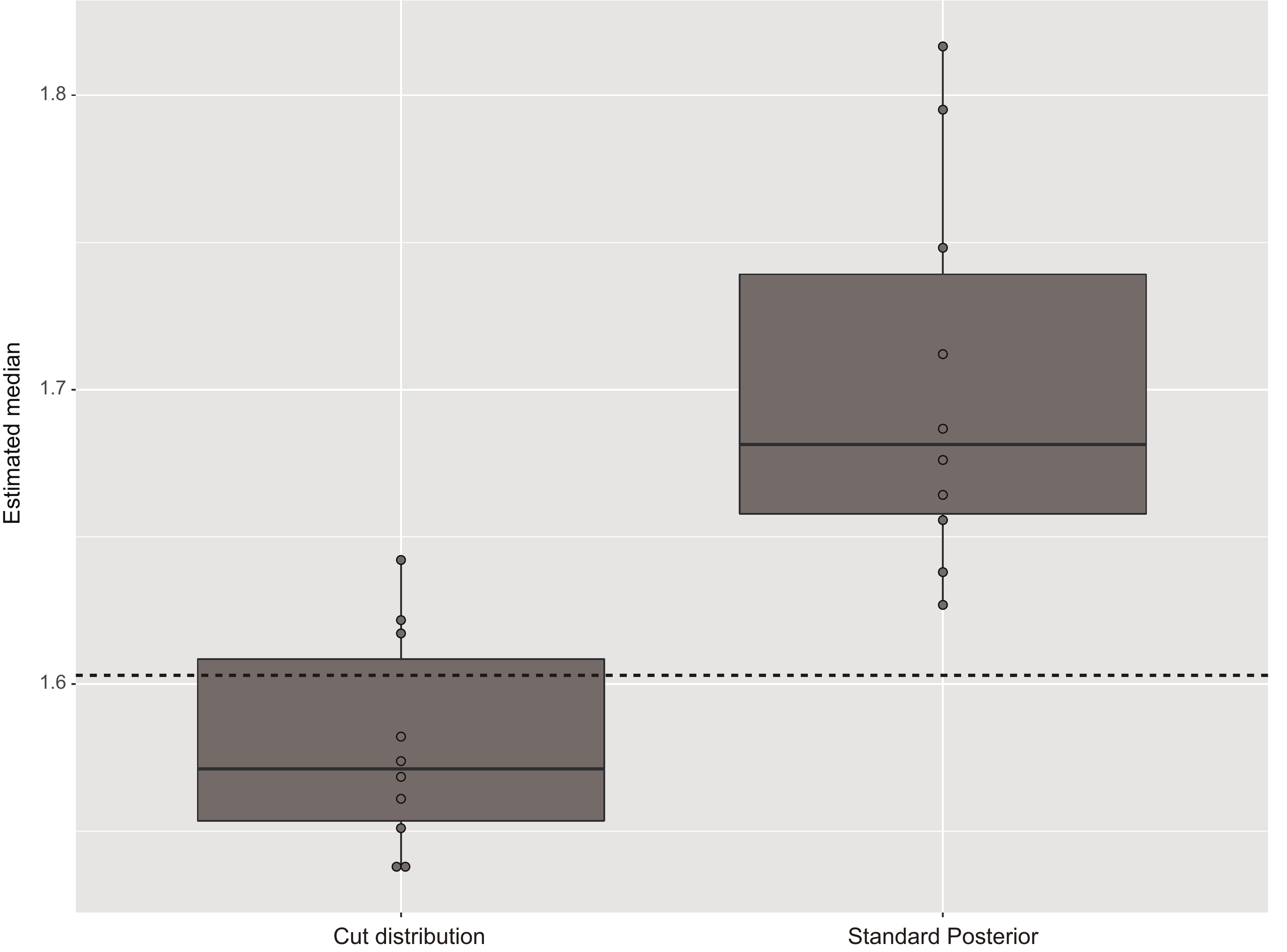}
\caption[Box plot of median estimates for $\varphi_1^2$ from each of 10 independent runs, under the cut distribution and the standard Bayesian posterior.]{Box plot of median estimates for $\varphi_1^2$ from each of 10 independent runs, under the cut distribution and the standard Bayesian posterior. The dashed line indicates the true value of $\varphi_1^2$.} 
\label{F6}
\end{figure}

\subsection{Simulated Strong Dependence Between $\theta$ and $\varphi$}
In this section, we apply our algorithm in a simulated setting that illustrates when nested MCMC \citep{Plummer2015} can perform poorly. Consider the case when the distribution of $\theta$ is highly dependent on $\varphi$. In this case, if the distance between successive values $\varphi^\prime$ and $\varphi$ is large in the external MCMC chain, the weight function may not be close to 1 and so the internal chain will typically require more iterations to reach convergence. This will be particularly problematic if the mixing time for the proposal distribution is large. 

To simulate this scenario, we consider a linear regression for outcomes $Y_i$, $i = 1,...,50$, in which the coefficient vector $\theta=(\theta_1,...,\theta_d)$ is closely related to the coefficient $\varphi$ for covariate $X_i=(X_{\theta,i},X_{\varphi,i})$. To assess the performance under small and moderate dimension of $\theta$, we consider $d=1$ and 20 in this illustration. As well as observations of the outcome $Y_i$ and the covariate $X_i$, we assume we have separate observations $Z_j$, $j=1,...,100$ related to the coefficient $\varphi$.
\begin{equation}
\begin{aligned}
&Y_i\sim \text{N}(\theta^\intercal X_{\theta,i} +\varphi X_{\varphi,i},3), \ i=1,...,50; \\
&Z_j\sim \text{N}(\varphi,1),\ j=1, ..., 100;
\end{aligned}
\label{E37}
\end{equation}
Suppose that we wish to estimate $\varphi$ solely on the basis of $Z=(Z_1,...,Z_{100})$, and so we cut the feedback from $Y=(Y_1,...,Y_{50})$ to $\varphi$.

We generate $Y$ and $Z$ according to \eqref{E37}, with $\varphi=1$ and $\theta_p=\sin(p)$, $p = 1,..., d$, and compare the results of Stochastic Approximation Cut Algorithm (SACut), naive SACut and nested MCMC with internal chain length $n_{int} =$ 1, 10, 200, 500, 1000, 1500 and 2000. Notably, nested MCMC with $n_{int} =1$ is the WinBUGS algorithm. The proposal distribution for each element of $\varphi$ is a normal distribution, centred at the previous value and with standard deviation 0.25; and the proposal distribution for $\theta$ used in the nested MCMC is a normal distribution, centred at the previous value and with standard deviation $10^{-5}$. The priors for both parameters are uniformly distributed within a compact domain. We set the shrink magnitude $n_0=2000$ and precision parameter $\kappa_p=4$, $p=1,...,20$. The SACut and naive SACut algorithms are processed in parallel on ten cores of Intel Xeon E7-8860 v3 CPU (2.2 GHz) and the (inherently serial) nested MCMC algorithm is processed on a single core. All algorithms were independently run 20 times and the results are the averages across runs. Each run consists $5\times 10^4$ iterations. We retain only every $10^{th}$ value after discarding the first $40\%$ samples as burn-in. 

To assess the performance of these algorithms, we compare their estimation of $\mathbb{E}(\theta)$, lag-1 auto-correlation of samples, the Gelman-Rubin diagnostic statistic $\hat{R}$ \citep{gelman1992} and the average time needed for the whole run. The precision of the estimation of $\theta$ is measured by the mean square error (MSE) across its $d$ (either 1 or 20) components. The convergence is evaluated by averaging the Gelman-Rubin diagnostic statistic of $d$ components.

Results are shown in Table \ref{T1}. The time required to run the nested MCMC algorithm increases as the length of the internal chain or dimension of $\theta$ increases, although the influence of dimension of $\theta$ is relatively small. In a low dimensional case ($d=1$), the time needed to run SACut and naive SACut are more than the time needed to run the WinBUGS algorithm and nested MCMC algorithm when the length of internal chain is less than 500, but both the MSE and the Gelman-Rubin statistic are lower when using the SACut algorithm. In particular, the bias of the WinBUGS algorithm is large. There is only trivial difference in bias between SACut and nested MCMC when $n_{int}\geq 1000$, but SACut is significantly faster than nested MCMC. In the higher dimensional case ($d=20$), both SACut and naive SACut significantly outperform the WinBUGS and nested MCMC algorithm in terms of MSE. Although the difference between SACut and nested MCMC with $n_{int}=1000$ is small, the Gelman-Rubin statistic of the nested MCMC is still larger than the threshold 1.2 suggested by \cite{doi:10.1080/10618600.1998.10474787}. The MCMC chains produced by the nested MCMC converge better and the bias is smaller when $n_{int}\geq 1500$, but the SACut algorithm still outperforms it according to MSE and Gelman-Rubin statistic, and takes less time. It is also clear that nested MCMC samples show very strong auto-correlation for both cases and thinning may not efficiently solve this issue \citep{https://doi.org/10.1111/j.2041-210X.2011.00131.x}; both SACut and naive SACut do not show any auto-correlation. We also note that the estimates provided by SACut and naive SACut are almost identical in practice. However, since the time needed for both algorithm is almost the same, providing the full approach with a more solid theoretical foundation is a valuable contribution to the computational statistics literature for the cut distribution.

\cite{2017arXiv170803625J} recently proposed an unbiased coupling algorithm which can sample from the cut distribution. It requires running coupled Markov chains where samples from each chain marginally target the same true distribution. The estimator is completely unbiased when two chains meet. Drawing samples from the cut distribution using the unbiased coupling algorithm typically involves two stages. In general the first stage involves running coupled chains for $\varphi$ until they meet. For each sampled $\varphi$, the second stage involves running another set of coupled chains for $\theta$ until they meet. Although the algorithm is unbiased, as illustrated in Section 4.2, 4.3 and the discussion of \cite{2017arXiv170803625J}, the number of iterations for coupled chains is determined by meeting times, which can be very large especially when the dimension of the parameter is high. As a comparison, we apply the unbiased coupled algorithm on this example by using the R package ``unbiasedmcmc'' provided by \cite{2017arXiv170803625J}. To simplify the implementation and computation of the unbiased coupling algorithm, we consider a simplified scenario with an informative conjugate prior for $\varphi$, meaning we can omit the first stage and instead directly draw $5\times 10^4$ samples from $p(\varphi|Z)$. This prior is normal with mean equal to the true value of $\varphi$. We then ran preliminary coupled chains for $\theta$ that target $p(\theta|Y,\varphi)$ given these samples of $\varphi$ so as to sample the meeting times. Over the $5\times 10^4$ independent runs, the $95\%$ and $99\%$ quantile of meeting times were 44 and 147 respectively when $d=1$. Although the majority of meeting times are relatively small, their $95\%$ and $99\%$ quantile were 3525 and 5442 respectively when $d=20$. To ensure that the total number of iterations covers the majority of meeting times, following \cite{2017arXiv170803625J}, we set the minimum number of iterations for each coupled chain to 10 times the $95\%$ quantile of meeting times. The algorithm was processed in parallel on the same ten cores as SACut and final result is shown in Table \ref{T1}. Notably, unlike the nested MCMC algorithm, the computational time of the unbiased coupling algorithm increases significantly when the dimension of $\theta$ increases because it takes more time for coupled chains to couple in high dimensional cases. In the low dimensional case ($d=1$), the unbiased coupling algorithm performs better according to all metrics. In the higher dimensional case ($d=20$), the unbiased coupling algorithm achieves similar MSE to the SACut algorithm, but it takes considerably more computation time than SACut, even though the unbiased coupling algorithm was been conducted under a simplified setting (i.e., no coupled chain for $\varphi$).

\begin{table}
\caption{Mean squared error (MSE), lag-1 Auto-correlation (in absolute value) $|\text{AC}|$, Gelman-Rubin statistic $\hat{R}$, and clock time for the Stochastic Approximation Cut Algorithm (SACut), Naive SACut algorithm, WinBUGS algorithm, the nested MCMC algorithm (with varying internal chain length $n_{int}$) and unbiased coupling algorithm. All values are means across 20 independent runs.}
\label{T1}       
\resizebox{\hsize}{!}{$\begin{tabular}{l l r r r r r}
\hline\noalign{\smallskip}
$d$ & Algorithm & $n_{int}$ & MSE $\times 10^3$  & $|\text{AC}|$ &\multicolumn{1}{c}{$\hat{R}$} & Time (mins)  \\
\noalign{\smallskip}\hline\noalign{\smallskip}
& SACut & - & 0.112 &0.019 & 1.00 & 311 \\
& Naive SACut & - & 0.114 &0.016 & 1.00 & 308 \\
& WinBUGS & 1 & 355.280 &0.999 & 308.78 & 1 \\
&Nested MCMC & $10$ & 217.907 &0.999 & 29.87 & 10 \\
&Nested MCMC & $200$ & 0.158 &0.997 & 1.74 & 182 \\
1 &Nested MCMC & $500$ & 0.138 &0.993 & 1.25 & 454 \\
&Nested MCMC & $1000$ & 0.109 &0.990 & 1.07 & 910 \\
&Nested MCMC & $1500$ & 0.113 &0.986 & 1.08 & 1349 \\
&Nested MCMC & $2000$ &  0.118 &0.981 & 1.05 & 1771 \\
& Unbiased Coupling & - & 0.114 &0.012 & 1.01 & 22 \\
&  &  &  & &  & \\
& SACut & - & 1.42 &0.009 & 1.00 & 1239 \\
& Naive SACut & - & 1.47 &0.002 &1.01 &  1219 \\
& WinBUGS & 1 & 16387.69 &0.999 & 209.55 & 2 \\
&Nested MCMC & $10$ & 12490.25 &0.999 & 22.73 & 11 \\
&Nested MCMC & $200$ & 249.18 &0.999 & 2.38 & 259 \\
20 &Nested MCMC & $500$ & 10.76 &0.997 & 1.33 & 517 \\
&Nested MCMC & $1000$ & 1.86 &0.994 & 1.22 & 1010 \\
&Nested MCMC & $1500$ & 1.69 &0.991 & 1.19 & 1515 \\
&Nested MCMC & $2000$ &  1.60 &0.988 & 1.11 & 2058 \\
& Unbiased Coupling & - & 1.36 &0.013 & 1.00 & 2030 \\
\noalign{\smallskip}\hline
\end{tabular}$}
\end{table}

\subsection{Epidemiological Example}
We now consider an epidemiological study of the relation between high risk human papillomaviru (HPV) prevalence and cervical cancer incidence \citep{Maucort-Boulch717}, which was previously discussed by \cite{Plummer2015}. In this study, age-stratified HPV prevalence data and cancer incidence data were collected from 13 cities. The model is divided into two modules. The first module concerns the number of people with HPV infection in city $i$, denoted as $Z_i$, out of a sample of $N_i$ women:
\[
Z_i\sim\text{Bin}(N_i,\varphi_i).
\]
The second module describes the relation between the number of cancer cases $Y_i$ from $T_i$ person-years and incidence which is assumed to be linked with $\varphi_i$ by a log linear relationship:
\[
Y_i \sim \text{Poisson}\left(T_i\left(\exp(\theta_1+\theta_2\varphi_i)\right)\right).
\]
The log-linear dose-response relationship is speculative, so we apply the cut algorithm to prevent the feedback from the second module to the estimation of $\varphi_i$ \citep{Plummer2015}.

We apply the Stochastic Approximation Cut Algorithm and compare results with the standard Bayesian approach (i.e. without a cut). Both algorithms were run 10 times independently, each with $1.4\times 10^5$ iterations. We set the shrink magnitude $n_0=20000$ and precision parameter $\kappa_1=3$ for $\theta_1$ and $\kappa_2=2$ for $\theta_2$. We retain only every $100^{th}$ value after discarding the first $4\times 10^4$ samples as burn-in. The pooled results of $\theta$ are shown in Figure \ref{F7},  highlighting the considerable effect of cutting feedback in this example. Our results are consistent with existing studies: specifically the scatter plot and density plot agree with \cite{jacob2017better} and \cite{2020arXiv200306804C}. Our results are also consistent with the results of nested MCMC algorithm when its internal chain length is largest (see \cite{Plummer2015}). This again shows that the SACut algorithm provides similar estimates to the nested MCMC algorithm with a large internal chain length.

\begin{figure}[tp] 
\centering 
\includegraphics[width=\textwidth]{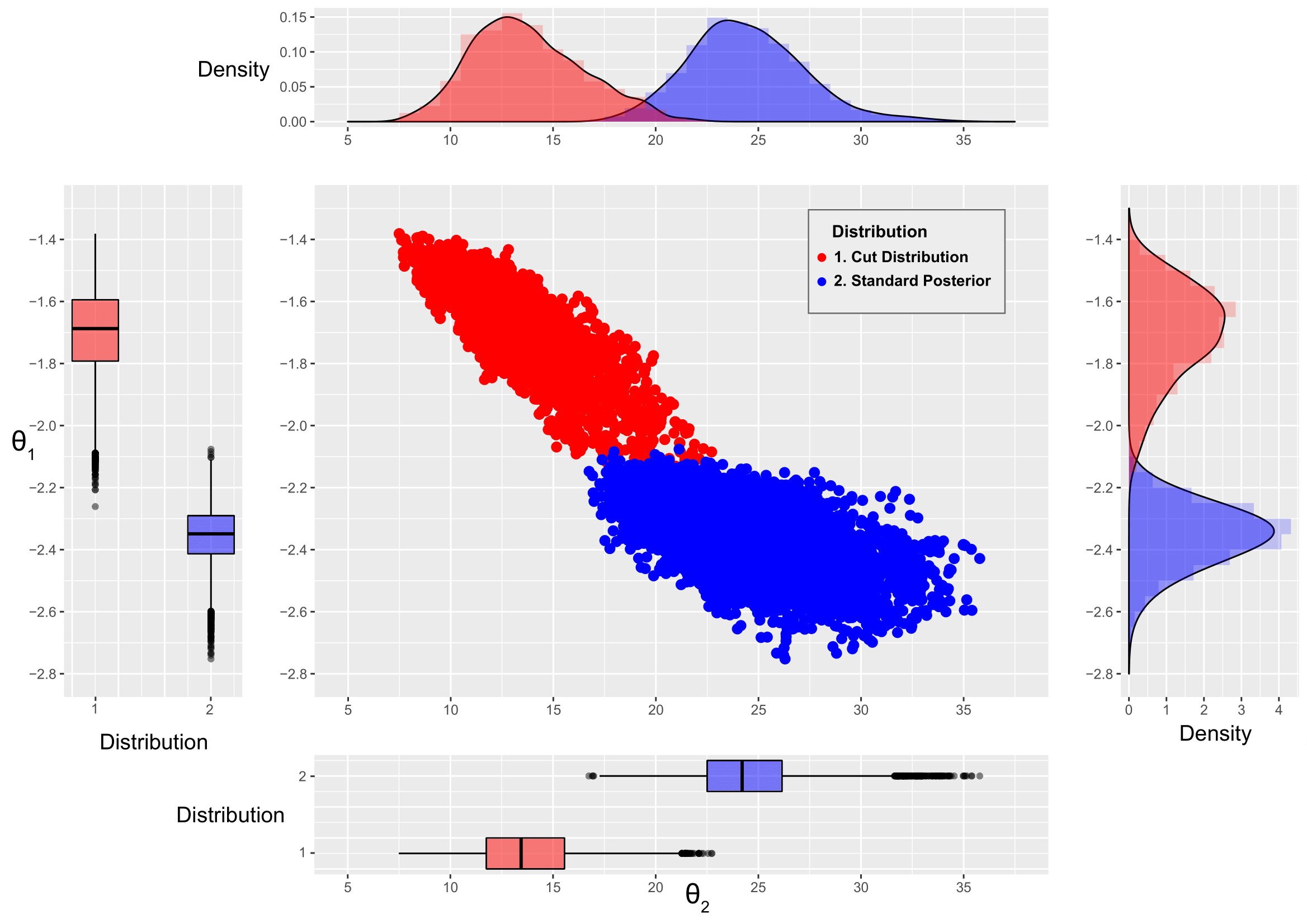} 
\caption[Comparison of the distribution of $\theta_1$ and $\theta_2$ drawn from the cut distribution (red) and standard Bayesian posterior (blue).]{Comparison of the distribution of $\theta_1$ and $\theta_2$ drawn from the cut distribution (red) and standard Bayesian posterior (blue).} 
\label{F7}
\end{figure}

\section{Conclusion}
\label{sec:conc}
We have proposed a new algorithm for approximating the cut distribution that improves on the WinBUGS algorithm and approximate approaches in \cite{Plummer2015}. Our approach approximates the intractable marginal likelihood $p(Y|\varphi)$ using Stochastic Approximation Monte Carlo \citep{doi:10.1198/016214506000001202}. The algorithm avoids the weakness of approximate approaches that insert an ``internal limit'' into each iteration of the main Markov chain. Obviously, one can argue that approximate approaches can be revised by setting the length of the internal chain to the number of iterations, i.e. $n_{int}=n$ so that the internal length diverges with $n$. However, since the sampling at each iteration is still not perfect and bias is inevitably introduced, the convergence of the main Markov chain remains unclear and the potential limit is not known. We proved convergence of the samples drawn by our algorithm and present the exact limit, though its convergence rate is not fully studied and needs further investigations. Although the bias is not completely removed by our algorithm, the degree of the bias is explicit in the sense that the shape of $p^{\kappa}(\theta|Y,\varphi)$ is known since the shape of $p(\theta|Y,\varphi)$ is normally obtainable given a fixed $\varphi$. Corollary \ref{coro2} shows that the bias in our approach can be reduced by increasing the precision parameter $\kappa$. We proposed that $\kappa$ be selected by comparing results across a range of choices; quantitative selection of this precision parameter still needs further study.

Existing approximate approaches \citep{Plummer2015} which need an infinitely long internal chain may be computationally slow, because the internal chain requires sequential calculation so parallelization is not possible. In contrast, thanks to the embarrassingly parallel calculation of \eqref{E12}, our algorithm can be more computationally efficient when multiple computer cores are available, although the per-iteration time of our algorithm decays as the Markov chain runs due to the increasing size of collection of auxiliary variables. 

Lastly, while the adaptive exchange algorithm \citep{doi:10.1080/01621459.2015.1009072} is used for intractable normalizing problems when the normalizing function is an integral with respect to the observed data, it would be interesting to investigate the use of our algorithm for other problems involving a normalizing function that is an integral with respect to the unknown parameter. For example, our algorithm can be directly extended to sample from the recently developed Semi-Modular Inference distribution \citep{2020arXiv200306804C} which generalizes the cut distribution. 

\section*{Supplementary Materials}
The supplementary appendix contains all technical proofs of results stated in the paper.

\section*{Acknowledgement}
The authors thank Daniela De Angelis and Simon R. White for helpful discussions and suggestions. Yang Liu was supported by a Cambridge International Scholarship from the Cambridge Commonwealth, European and International Trust. Robert J.B. Goudie was funded by the UK Medical Research Council [programme code MC\textunderscore UU\textunderscore 00002/2]. The views expressed are those of the authors and not necessarily those of the NIHR or the Department of Health and Social Care.



\appendix
\appendixpage
\addappheadtotoc
\section{The construction of the auxiliary parameter set $\Phi_0$}
According to \cite{doi:10.1080/01621459.2015.1009072}, in order to have a good auxiliary parameter set $\Phi_0$ so that the set of $p(\theta|Y,\varphi)$, $\varphi\in\Phi_0$, reasonably reflects the truth $p(\theta|Y,\varphi)$, where $\varphi\sim p(\varphi|Z)$, two conditions should be satisfied.

\begin{itemize}
\item Full representation: Let $C_{\Phi_0}$ be the convex hull constructed from $\Phi_0$, then we require $\int_{C_{\Phi_0}}p(\varphi|Z)d\varphi \approx 1$. This ensures that the selected $\Phi_0$ has fully represented the original domain of $\varphi$.

\item Reasonable overlap: For neighbouring $\varphi_0^{(i)}$ and $\varphi_0^{(j)}$, we require $p(\theta|Y,\varphi_0^{(i)})$ and $p(\theta|Y,\varphi_0^{(j)})$ should have a reasonable overlap. Hence, the probability of accepting new proposal of $\varphi\in\Phi_0$ is reasonably large given a fixed $\theta$ for the auxiliary chain. This ensures that the auxiliary chain can mix well.
\end{itemize}

The grid $\Phi_0$ is chosen by following the Max-Min procedure \citep{doi:10.1080/01621459.2015.1009072} and the purpose is to use this discrete set as a representation of the domain of $p(\varphi|Z)$. After we have decided the number of auxiliary parameters (that is $m$), the auxiliary parameter set $\Phi_0$ is formed by selecting $\varphi$ from a larger set $\Phi^{(M)}=\{\varphi^{(1)},\cdots,\varphi^{(M)}\}$ (an arbitrary but substantially larger $M>m$) which are drawn from the marginal posterior $p(\varphi|Z)$ by any standard MCMC algorithm. Before starting the iterative process, we standardize all 
$\varphi_0^{(i)}, i=1,\cdots,M$ (i.e, $\varphi_{new} = (\varphi - \varphi_{min})/ (\varphi_{max} - \varphi_{min})$). We then add values to the grid using the following iterative process, initialised by randomly selecting a $\varphi_0^{(1)}$ as the first step. Then suppose at the $k^{th}$ step we have $\Phi_0^{(k)}=\{\varphi_0^{(1)},\cdots,\varphi_0^{(k)}\}$. For each $\varphi\in \Phi^{(M)}$ that has not yet been selected, we calculate the minimum distance to the set $\Phi_0^{(k)}$ according to some pre-defined distance measure (e.g., Euclidean distance). This is the `Min' process. We then find the $\varphi$ that has not been selected but has the maximum distance to the set $\Phi_0^{(k)}$. This is the `Max' process. We then add this particular $\varphi$ into $\Phi_0^{(k)}$ to form $\Phi_0^{(k+1)} = \Phi_0^{(k)} \cup \{\varphi\}$.

Here we use a toy example to illustrate the procedure of selecting $m$. To make clear visualization of $\varphi$ straightforward,  we use the example derived from Section 4.1 in the main text but reduce the dimension of $\varphi$ to 2. 

We first draw a large number ($M$) of samples of $\varphi$ from its marginal posterior $p(\varphi|Z)$ by a standard MCMC method and pool all samples together as a benchmark sample set. This step is feasible because $p(\varphi|Z)$ is a standard posterior distribution so there is no double intractability. If $M$ is large enough, it is appropriate to regard the convex hull of the benchmark sample set as a good approximation of the true domain of $p(\varphi|Z)$.

Next, we select several candidates values of $m$. In this illustration, we simply select $m$ from $\{10,20,50,100,500\}$. Given a selected value of $m$, we use the Max-Min procedure \citep{doi:10.1080/01621459.2015.1009072} to construct the $\Phi_0^{(m)}$, and calculate its corresponding convex hull $C_{\Phi_0^{(m)}}$ numerically using the R package \textit{geometry}. See the attached Figure \ref{FF1}. It can be seen that the shape of the convex hull approximates the benchmark convex hull increasingly accurately as $m$ increases. 

\begin{figure}[thp]
\centering 
\includegraphics[angle=270,scale=0.5]{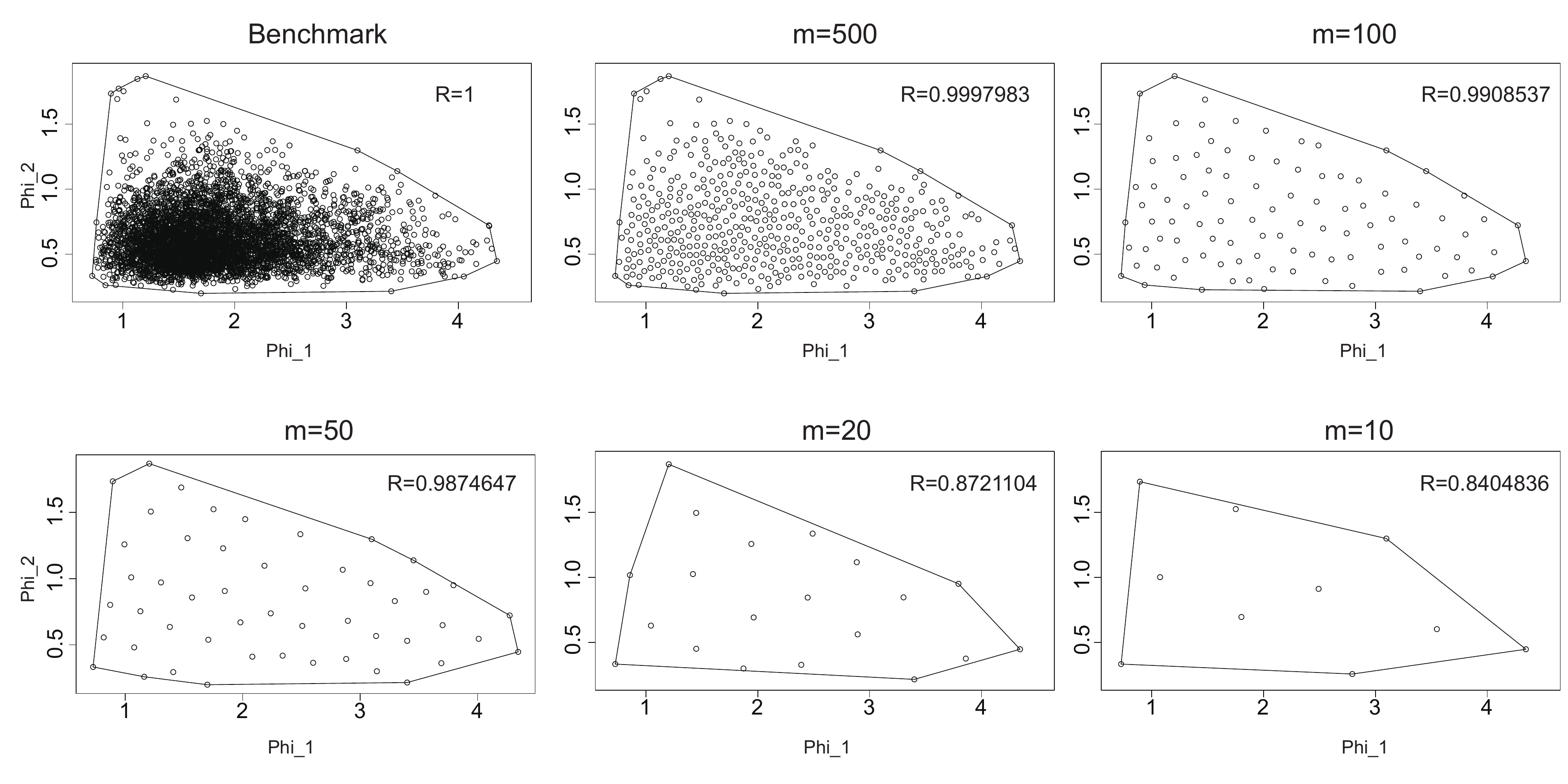}
\caption[Convex hull]{Convex hull $C_{\Phi_0}$ (indicated by the black circle) and samples of $\varphi$ (indicated by dots) when $m=$ 10, 20, 50, 100, 500 and Benchmark. The coverage ratio is shown on the upper right corner of each subplot.} 
\label{FF1}
\end{figure}

Let $C_{\Phi_0^{(M)}}$ be the benchmark convex hull, and define the coverage ratio $R_m$ by:
\begin{equation}
R_m = \frac{\int_{C_{\Phi_0^{(m)}}}p(\varphi|Z)d\varphi}{\int_{C_{\Phi_0^{(M)}}}p(\varphi|Z)d\varphi},
\end{equation}
where $R_m$ can be numerically approximated using Monte Carlo. The $m$ can be selected by choosing the smallest $m$ that gives a high coverage ratio $R_m$ (e.g., $R_m > 0.95$). In this example, since the the coverage ratio is $>0.98$ we choose $m=50$ to satisfy the first condition (Full representation).

We then check whether this value of $m$ also satisfies the second condition (Reasonable overlap). This involves checking that the $m$ different $p(\theta|Y,\varphi_0^{(i)})$, where $\varphi_0^{(i)}\in \Phi_0^{(m)}$, overlap sufficiently. To do this, for each $\varphi_0^{(i)}$, we draw samples of $\theta$ using a standard MCMC method and compare the empirical distribution of $\theta\sim p(\theta|Y,\varphi_0^{(i)})$ with the empirical distribution of $\theta\sim p(\theta|Y,\varphi_0^{(i1)})$ and $p(\theta|Y,\varphi_0^{(i2)})$, where $\varphi_0^{(i1)}$ and $\varphi_0^{(i2)}$ are the closest and second closest values in $\Phi_0^{(m)}$ to $\varphi_0^{(i)}$ in terms of Euclidean distance. This can be visually shown as a grouped box-plot in Figure \ref{FF2}. It is clear that the majority part (inter-quartile area) of the empirical distribution of $\theta$ given each $\varphi_0^{(i)}\in \Phi_0^{(m)}$ overlaps inter-quartile area given its neighbouring $\varphi_0^{(i1)}$ and $\varphi_0^{(i2)}$. Hence, it is appropriate to argue that $m=50$ satisfies the second condition (Reasonable overlap). 

\begin{figure}[thp]
\centering 
\includegraphics[scale=0.95]{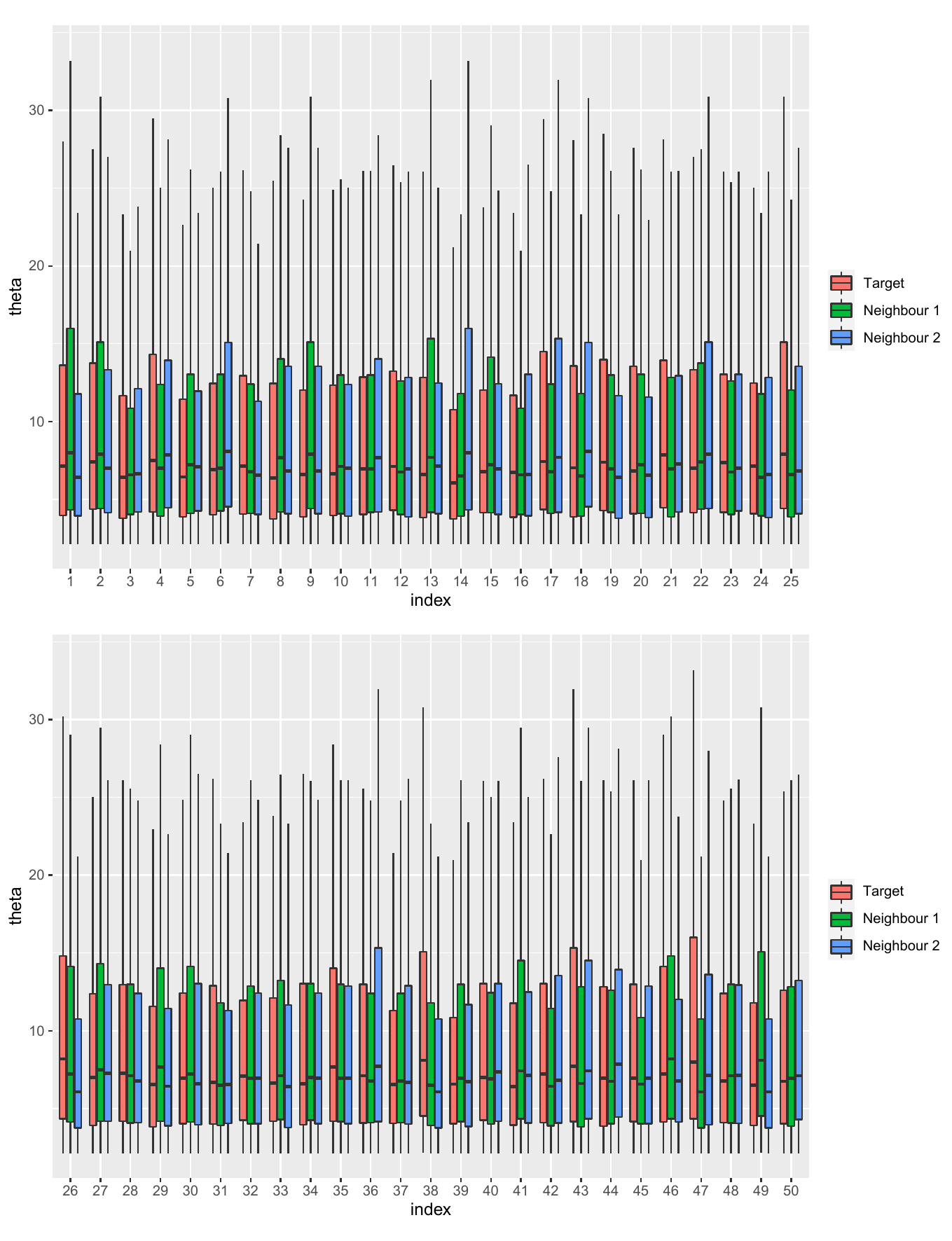}
\caption[Boxplot]{Boxplot of the empirical distribution of $\theta$ given $\varphi$ (i.e., $p(\theta|Y,\varphi)$). For an arbitrary fixed index $i$, the red boxplot refers to $\varphi_0^{(i)}$. The green boxplot refers to $\varphi_0^{(i1)}$ (the closest neighbour). The blue boxplot refers to $\varphi_0^{(i2)}$ (the second closest neighbour).} 
\label{FF2}
\end{figure}

In general, $m$ grows with the dimension of $\varphi$. However, the exact relationship between them depends on the context of the real problem. Also, a large $m$ is not a necessary condition for proposed theorems to be theoretically valid. This is because $m$ is involved in the construction of the proposal distribution for the importance sampling procedure that forms the numerator and denominator of $P_n^\ast(\theta|Y,\varphi)$ and we could use any proposal distribution only if it has a correct domain, although a proposal distribution based on a smaller $m$ could leads to a slow convergence of the auxiliary chain. In the special case when the marginal cut distribution $p(\theta|Y,\varphi)$ is not sensitive to the change of $\varphi$, a small $m$ might be good enough. A larger $m$ will bring fewer benefits when two conditions hold in practice. This differs significantly from the fact that we always require $n$ to go to infinity. Hence, increasing $m$ has a diminishing marginal utility after two conditions hold. Notwithstanding, when the dimension of $\varphi$ increases, it will be more difficult to check whether two conditions hold because the numerical calculation of the convex hull will become extremely time-consuming and also checking for overlap visually will be infeasible. A more practical way could be simply checking whether the auxiliary chain can converge well given a reasonable time period by running some preliminary trials as suggested in \cite{doi:10.1080/01621459.2015.1009072}. 

\section{Construction of $P_n^\ast(\theta|Y,\varphi)$}
Given a particular $\varphi^\prime$, measure $P_n^\ast(\theta|Y,\varphi)$ is actually a weighted average of $\mathbbm{1}_{\{\theta\in\mathcal{B}\}}$. It is used to approximate $\int_\mathcal{B} p(Y|\theta,\varphi^\prime)p(\theta)d\theta/p(Y|\varphi^\prime)$, where numerator and denominator are separately approximated by dynamic importance sampling. The only difference is the domain of integration (i.e., $\mathcal{B}$ versus $\Theta$). Here, we only show the numerator. We have
\[
\int_\mathcal{B} p(Y|\theta,\varphi^\prime)p(\theta) d\theta \propto \mathbb{E}_{\theta}\left(\mathbbm{1}_{\{\theta\in\mathcal{B}\}}\right),\ \textbf{where}\ \theta\sim \frac{p(Y|\theta,\varphi^\prime)p(\theta)}{K(\varphi^\prime)},
\]
where $K(\varphi^\prime)$ is a intractable normalizing constant of the target distribution. Hence, we cannot sample $\theta$ from $p(Y|\theta,\varphi^\prime)p(\theta)/K(\varphi^\prime)$. We have shown in the main text that, when iteration number $j$ is large enough, we actually sample $(\tilde{\theta},\tilde{\varphi})$ from an iteration-specific proposal distribution
\[
p_{j}(\tilde{\theta},\tilde{\varphi})\propto\sum_{i=1}^m \frac{p(Y|\tilde{\theta},\varphi_0^{(i)})p(\tilde{\theta})}{\tilde{w}_{j-1}^{(i)}}\mathbbm{1}_{\{\tilde{\varphi}=\varphi_0^{(i)}\}},\ \tilde{\theta}\in\Theta,\ \tilde{\varphi}\in\Phi_0.
\]
Hence, we have
\[
\begin{aligned}
&\mathbb{E}_{\theta}\left(\mathbbm{1}_{\{\theta\in\mathcal{B}\}}\right) = \mathbb{E}_{\tilde{\theta},\tilde{\varphi}}\left(\mathbbm{1}_{\{\tilde{\theta}\in\mathcal{B}\}}\frac{p(Y|\tilde{\theta},\tilde{\varphi}^\prime)p(\tilde{\theta})}{p_{j}(\tilde{\theta},\tilde{\varphi})K(\varphi^\prime)}\right) \\
&\ =\mathbb{E}_{\tilde{\theta},\tilde{\varphi}}\left\{\mathbbm{1}_{\{\tilde{\theta}\in\mathcal{B}\}} \frac{p(Y|\tilde{\theta},\varphi^\prime)p(\tilde{\theta})}{K(\varphi^\prime)} \left(\sum_{i=1}^m  \frac{\tilde{w}_{j-1}^{(i)}}{p(Y|\tilde{\theta},\varphi_0^{(i)})p(\tilde{\theta})}\mathbbm{1}_{\{\varphi_0^{(i)}=\tilde{\varphi}\}} \right)\right\},\ \textbf{where}\ (\tilde{\theta},\tilde{\varphi})\sim p_{j}(\tilde{\theta},\tilde{\varphi}).
\end{aligned}
\]
The above expectation can be approximated by a step $j$ single sample Monte Carlo estimator
\[
\frac{1}{K(\varphi^\prime)}\sum_{i=1}^m \tilde{w}_{j-1}^{(i)} \frac{p(Y|\tilde{\theta}_j,\varphi^\prime)}{p(Y|\tilde{\theta}_j,\varphi_0^{(i)})} \mathbbm{1}_{\{\tilde{\theta}_j\in\mathcal{B},\varphi_0^{(i)}=\tilde{\varphi}_j\}},\ \textbf{where}\ (\tilde{\theta}_j,\tilde{\varphi}_j)\sim p_{j}(\tilde{\theta},\tilde{\varphi}).
\]
Note that, $K(\varphi^\prime)$ will be canceled out in the denominator and numerator of measure (5) so we can omit it. Now we have a total of $n$ samples $\{(\tilde{\theta}_j,\tilde{\varphi}_j)\}_{j=1}^n$, we then add all single sample Monte Carlo estimators and calculate the average
\[
\frac{1}{n}\sum_{j=1}^n\sum_{i=1}^m \tilde{w}_{j-1}^{(i)} \frac{p(Y|\tilde{\theta}_j,\varphi^\prime)}{p(Y|\tilde{\theta}_j,\varphi_0^{(i)})} \mathbbm{1}_{\{\tilde{\theta}_j\in\mathcal{B},\varphi_0^{(i)}=\tilde{\varphi}_j\}},\ \textbf{where}\ (\tilde{\theta}_j,\tilde{\varphi}_j)\sim p_{j}(\tilde{\theta},\tilde{\varphi}).
\]
Similarly, $1/n$ will be canceled out.

\section{Density Function Approximation by Simple Function}
Here, we show how a density function $f$ can be approximated by a simple function that is constant on a hypercube. We show that the degree of approximation can be easily controlled, and is dependent on the gradient of $f$. The use of a simple function to approximate a density function has been discussed previously \citep{Fu2002, doi:10.1080/03610910802657904}, but here we use a different partition of the support of the function, determined by rounding to a user-specified number of decimal places. 

For a compact set $\Psi\subset \mathbb{R}^d$ with dimension $d$, define a map $\mathcal{R}_\kappa:\ \Psi\rightarrow\Psi$ that rounds every element of $\psi\in\Psi$ to $\kappa$ decimal places, where $\kappa\in\mathbb{Z}$, as $\mathcal{R}_\kappa(\psi)=\lfloor10^\kappa\psi+0.5\rfloor/10^\kappa.$ Since $\Psi$ is compact, $\mathcal{R}_\kappa(\Psi)$ is a finite set and we let $R_\kappa$ denote its cardinality. We partition $\Psi$ in terms of (partial) hypercubes $\Psi_r$ whose centres $\psi_r\in\mathcal{R}_\kappa(\Psi)$ are the rounded elements of $\Psi$,
\begin{equation}
\Psi_r=\Psi\cap\{\psi:\left\|\psi-\psi_r\right\|_\infty\leq 5\times10^{-\kappa-1}\},\ r=1,...,R_\kappa,
\label{EE15}
\end{equation}
and the boundary set $\bar{\Psi}_\kappa$,
\begin{equation}
\bar{\Psi}_\kappa=\Psi\cap\left(\bigcup_{r=1}^{R_\kappa} \{\psi:\left\|\psi-\psi_r\right\|_\infty= 5\times10^{-\kappa-1}\}\right).
\label{EE16}
\end{equation}
It is clear that $\bigcup_{r=1}^{R_\kappa}\Psi_r=\Psi$. Hence $\{\Psi_r\backslash \bar{\Psi}_\kappa\}_{r=1}^{R_\kappa}$ and $\bar{\Psi}_\kappa$ form a partition of $\Psi$.

Using this partition, we are able to construct a simple function density that approximates a density function. Let $\mathcal{C}$ be the set of all continuous and integrable probability density functions $f:\Psi\rightarrow \mathbb{R}$, and let $\mathcal{S}$ be the set of all simple functions $f:\Psi \rightarrow \mathbb{R}$. Define a map $\mathcal{S}_\kappa:\ \mathcal{C}\rightarrow \mathcal{S}$ for any $f\in \mathcal{C}$ as
\[
\mathcal{S}_\kappa(f)(\psi)=\sum_{r=1}^{R_\kappa}\frac{1}{\mu(\Psi_r)}\int_{\Psi_r}f(\psi^\prime)d\psi^\prime \mathbbm{1}_{\{\psi\in\Psi_r\}},\ \forall \psi\in\Psi .
\]
The sets $\Psi_r$, $r=1,...,R_\kappa$, are the level sets of the simple function approximation, and the value $\mathcal{S}_\kappa(f)(\psi)$, $\psi\in\Psi\backslash \bar{\Psi}_\kappa$, is the (normalized) probability of a random variable with density $f$ taking a value in $\Psi_r$, $r=1,...,R_\kappa$. Note that, when $\Psi_r$ is a full hypercube, $\mu(\Psi_r)=10^{-d\kappa}$; and if the set $\Psi$ is known, then $\mu(\Psi_r)$ is obtainable for partial hypercubes. Figure \ref{FF3} illustrates how this simple function approximates the truncated standard normal density function $f_{\text{norm}}:\ [-4,4]\rightarrow\mathbb{R}$, when $\kappa=0$ and $\kappa=1$. Note that this is the optimal simple function for the approximation in terms of  Kullback-Leibler  divergence \citep{doi:10.1080/03610910802657904}.

\begin{figure}[t] 
\centering 
\includegraphics[scale=0.3]{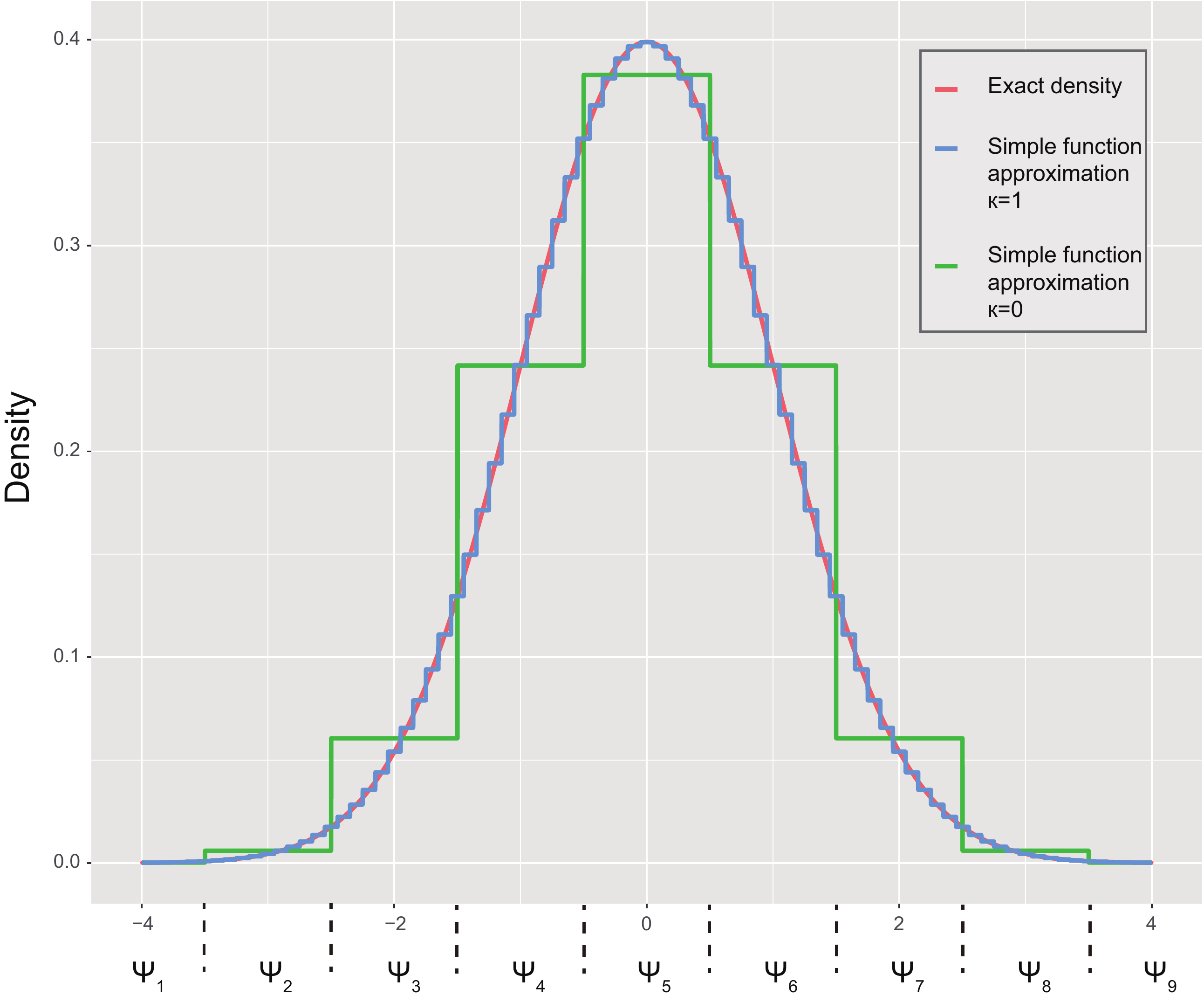} 
\caption[Simple function approximation of a truncated normal distribution.]{Simple function approximation of a truncated normal distribution. When $\kappa=0$ the sets $\Psi_1=[-4,-3.5]$, $\Psi_2=[-3.5,-2.5]$,..., $\Psi_{8}=[2.5,3.5]$, $\Psi_{9}=[3.5,4]$ are the intervals partitioning $[-4,4]$ and $\bar{\Psi}_0=\{-3.5,-2.5,...,2.5,3.5\}$.} 
\label{FF3}
\end{figure}

Since $\mu(\bar{\Psi}_\kappa)=0$, it is clear that
\[
\int_{\Psi}\mathcal{S}_\kappa(f)(\psi)d\psi=\int_{\Psi}f(\psi)d\psi=1.
\]
Hence, $\mathcal{S}_\kappa(f)$ is a well-defined density function. We have the following theorem.

\begin{theorem}
\label{them1}
Given any continuous density function $f$,
\[
\mathcal{S}_\kappa(f)\xrightarrow{\text{a.s.}}f,\ \ \ \text{as}\ \kappa\rightarrow\infty.
\]
\end{theorem}
\begin{proof}
Let $\mathbb{Q}$ be the set of all rational numbers in $\mathbb{R}$ and hence $\mathbb{Q}^c$ is the set of all irrational numbers in $\mathbb{R}$. Let $\mathscr{E}={\mathbb{Q}^c}^d\cap\Psi$ and it is easy to see that $\mu(\mathscr{E})=\mu(\Psi)$ since $\mu(\mathbb{Q})=0$. We first show that, $\forall \kappa<\infty$ and $\forall \psi\in\mathscr{E}$, we have $\psi\notin\bar{\Psi}_\kappa$.

Given a $\kappa<\infty$, every element of set $\mathcal{R}_\kappa(\Psi)$ is a $d$-dimensional rational vector. We also have that $5\times10^{-\kappa-1}$ is a rational number. Therefore, at least one element of $d$-dimensional vector $\psi$ is a rational number if $\psi\in\bigcup_{r=1}^{R_\kappa} \{\psi:\left\|\psi-\psi_r\right\|_\infty= 5\times10^{-\kappa-1}\}$. Now $\forall \psi\in\mathscr{E}$, because $\psi$ is a $d$-dimensional irrational vector, $\psi\notin\bigcup_{r=1}^{R_\kappa} \{\psi:\left\|\psi-\psi_r\right\|_\infty= 5\times10^{-\kappa-1}\}$, and hence $\psi\notin\bar{\Psi}_\kappa$.

Now given a fixed $\kappa<\infty$, $\forall \psi\in\mathscr{E}$, since $\psi\notin\bar{\Psi}_\kappa$, $\psi$ is always in the inner set of one of $\Psi_r$, $r=1,...,R_\kappa$. Re-write this $\Psi_r$ as $\Psi_\psi^{(\kappa)}$. Since the set $\Psi_\psi^{(\kappa)}$ is compact and function $f$ is continuous, we have $f_{\psi,min}=\min_{y\in\Psi_\psi^{(\kappa)}}f(y)$ and $f_{\psi,max}=\max_{y\in\Psi_\psi^{(\kappa)}}f(y)$. By the first mean value theorem, there is a $\psi^\ast\in\Psi_\psi^{(\kappa)}$ with $f_{\psi,min}\leq f(\psi^\ast)\leq f_{\psi,max}$, such that
\[
\mathcal{S}_\kappa(f)(\psi)=\frac{1}{\mu(\Psi_\psi^{(\kappa)})}\int_{\Psi_\psi^{(\kappa)}}f(y)dy=f(\psi^\ast)\frac{1}{\mu(\Psi_\psi^{(\kappa)})}\int_{\Psi_\psi^{(\kappa)}}dy=f(\psi^\ast).
\]

It is clear that, when $\kappa$ increases, $\mu(\Psi_\psi^{(\kappa)})$ monotonically decreases since $\Psi_\psi^{(\kappa+1)}\subset\Psi_\psi^{(\kappa)}$ (i.e. a much smaller hypercube is formed). This leads to the fact that $(f_{\psi,max}-f_{\psi,min})$ monotonically decreases to 0. Hence, there is a $N$ such that $\forall \kappa>N$, $(f_{\psi,max}-f_{\psi,min})\leq\varepsilon$. Then we have $\forall \kappa>N$,
\[
|\mathcal{R}_\kappa^\ast(f)(\psi)-f(\psi)|=|f(\psi^\ast)-f(\psi)|\leq (f_{\psi,max}-f_{\psi,min})\leq\varepsilon.
\]
Hence,
\[
\mathcal{S}_\kappa(f)\xrightarrow{\text{a.s.}}f,\ \ \ \text{as}\ \kappa\rightarrow\infty.
\]
\end{proof}

When the density function $f$ is also continuously differentiable, we can obtain the following result about the rate of convergence. 

\begin{corollary}
\label{scoro1}
Given a density function $f$ that is continuously differentiable, there exists a set $\mathscr{E}\subset\Psi$ with $\mu(\mathscr{E})=\mu(\Psi)$ such that the local convergence holds:
\[
|\mathcal{S}_\kappa(f)(\psi)-f(\psi)|\leq (\varepsilon(\psi,\kappa)+\left\|\nabla f(\psi)\right\|_2)\frac{\sqrt{d}}{10^\kappa},\ \forall \psi\in\mathscr{E},
\] 
where $\varepsilon(\psi,\kappa)\rightarrow 0$ as $\kappa\rightarrow\infty$.

In addition, the global convergence holds:
\[
\sup_{\psi\in\mathscr{E}}|\mathcal{S}_\kappa(f)(\psi)-f(\psi)|\leq \sup_{\psi\in\Psi}\left\|\nabla f(\psi)\right\|_2 \frac{\sqrt{d}}{10^\kappa}.
\]
\end{corollary}
\begin{proof}
Following the result of Theorem 1, for a given $\psi\in\mathscr{E}$, we have
\[
|\mathcal{S}_\kappa(f)(\psi)-f(\psi)|\leq (f_{\psi,max}-f_{\psi,min}).
\]
Since $f$ has a continuous gradient on a compact set, then by the mean value theorem we have:
\[
(f_{\psi,max}-f_{\psi,min})=|\langle\nabla f(y), (\psi_{max}-\psi_{min})\rangle|.
\]
where $\langle\cdot,\cdot\rangle$ means inner product, $f(\psi_{max})=f_{\psi,max}$, $f(\psi_{min})=f_{\psi,min}$, $y\in\Psi_\psi^{(\kappa)}$.
By the Cauchy-Schwarz inequality, we have
\[
|\langle\nabla f(y), (\psi_{max}-\psi_{min})\rangle|\leq \left\|\nabla f(y)\right\|_2\times \left\|(\psi_{max}-\psi_{min})\right\|_2
\]

Now we prove the local convergence result. Since $\nabla f$ is continuous on the $d$-dimensional compact set $\Psi$, we can write
\[
\varepsilon(\psi,\kappa)=\sup_{a,b\in\Psi_\psi^{(\kappa)}}\left\|\nabla f(a)-\nabla f(b)\right\|_2.
\]
Since $\mu(\Psi_\psi^{(\kappa)})\rightarrow 0$, it is easy to check that $\varepsilon(\psi,\kappa)\rightarrow 0$ when $\kappa\rightarrow\infty$. Moreover, we have both $\psi_{max}$ and $\psi_{min}$ are in set $\Psi_\psi^{(\kappa)}$, and we have 
\[
\sup_{a,b\in\Psi_\psi^{(\kappa)}}\left\|(a-b)\right\|_2=\sqrt{d10^{-2\kappa}}.
\]
Then by the triangle inequality, we have
\[
\begin{aligned}
\left\|\nabla f(y)\right\|_2\times \left\|(\psi_{max}-\psi_{min})\right\|_2 &\leq \left(\left\|\nabla f(y)-\nabla f(\psi)\right\|_2+\left\|\nabla f(\psi)\right\|_2\right)\frac{\sqrt{d}}{10^\kappa} \\
&\leq (\varepsilon(\psi,\kappa)+\left\|\nabla f(\psi)\right\|_2)\frac{\sqrt{d}}{10^\kappa}.
\end{aligned}
\]
and hence
\[
|\mathcal{S}_\kappa(f)(\psi)-f(\psi)|\leq (\varepsilon(\psi,\kappa)+\left\|\nabla f(\psi)\right\|_2)\frac{\sqrt{d}}{10^\kappa}.
\]

Now we prove the global convergence result. Since $\nabla f$ is continuous on compact set $\Psi$, then $\left\|\nabla f\right\|_2$ is bounded. We have
\[
\left\|\nabla f(y)\right\|_2\times \left\|(\psi_{max}-\psi_{min})\right\|_2 \leq \sup_{\psi\in\Psi}\left\|\nabla f(\psi)\right\|_2 \frac{\sqrt{d}}{10^\kappa}.
\]
Therefore, we have
\[
|\mathcal{S}_\kappa(f)(\psi)-f(\psi)|\leq \sup_{\psi\in\Psi}\left\|\nabla f(\psi)\right\|_2 \frac{\sqrt{d}}{10^\kappa}.
\]
Note that, this means that $|\mathcal{S}_\kappa(f)(\psi)-f(\psi)|$ is uniformly bounded. Hence, it implies
\[
\sup_{\psi\in\mathscr{E}}|\mathcal{S}_\kappa(f)(\psi)-f(\psi)|\leq \sup_{\psi\in\Psi}\left\|\nabla f(\psi)\right\|_2 \frac{\sqrt{d}}{10^\kappa}.
\]
\end{proof}

Corollary \ref{scoro1} shows that the rate of convergence of $\mathcal{S}_\kappa(f)$ to $f$ is geometric. It states that, (a) for any $\psi\in\mathscr{E}$, the rate of convergence is locally controlled by its gradient $\left\|\nabla f(\psi)\right\|_2$; and (b) the rate of convergence is uniformly controlled by the upper bound of the gradient. Hence, as is intuitively expected, convergence is faster if the target function $f$ has a smaller total variation on the set $\mathscr{E}$.

\begin{remark}
When the scale of each component of $\psi\in\Psi$ is not same, a more complex partition can be formed by choosing component-specific precision parameters $\kappa=(\kappa_1,...,\kappa_d)$. Denote $\circ$ as the Hadamard product and $10^{\pm\kappa}:=(10^{\pm\kappa_1},...,10^{\pm\kappa_d})$, we redefine
\[
\mathcal{R}_\kappa(\Psi)=\lfloor10^\kappa\circ\psi+0.5\rfloor\circ 10^{-\kappa}.
\]
We build a (partial) $d$-orthotope around $\psi_r\in\mathcal{R}_\kappa(\Psi)$
\[
\Psi_r=\Psi\cap\{\psi: |\psi-\psi_r|\leqq 5\times10^{-\kappa-1}\},\ r=1,...,R_\kappa.
\]
We do not discuss this more complex partition but all results in this paper that are based on the basic partition in \eqref{EE15} and \eqref{EE16} can be easily extended to this more complex partition.
\end{remark}

\section{Proofs of the Main Text}

\subsection{Proof of Lemma 1}
We write the explicit form of $p^{(\kappa)}(\theta|Y,\varphi)$:
\[
p^{(\kappa)}(\theta|Y,\varphi)=\mathcal{S}_\kappa(p(\cdot|Y,\varphi))(\theta)=\sum_{r=1}^{R_\kappa}\frac{1}{\mu(\Theta_r)}\int_{\Theta_r}p(\theta^\ast|Y,\varphi)d\theta^\ast\mathbbm{1}_{\{\theta\in\Theta_r\}},
\]
then we have:
\[
\begin{aligned}
&\sup_{\theta\in\Theta\setminus\bar{\Theta}_\kappa ,\varphi\in\Phi}\left|p_n^{(\kappa)}(\theta|Y,\varphi)-p^{(\kappa)}(\theta|Y,\varphi)\right| \\
&=\sup_{\theta\in\Theta\setminus\bar{\Theta}_\kappa ,\varphi\in\Phi}\left| \sum_{r=1}^{R_\kappa}\frac{1}{\mu(\Theta_r)}\left(W_n(\Theta_r|Y,\varphi)-\int_{\Theta_r}p(\theta^\ast|Y,\varphi)d\theta^\ast\right)\mathbbm{1}_{\{\theta\in\Theta_r\}}\right| \\
&\leq \sup_{\theta\in\Theta\setminus\bar{\Theta}_\kappa ,\varphi\in\Phi}\sum_{r=1}^{R_\kappa}\frac{1}{\mu(\Theta_r)}\left|W_n(\Theta_r|Y,\varphi)-\int_{\Theta_r}p(\theta^\ast|Y,\varphi)d\theta^\ast\right|\mathbbm{1}_{\{\theta\in\Theta_r\}} \\
&= \sup_{\varphi\in\Phi;1\leq r\leq R_\kappa}\frac{1}{\mu(\Theta_r)}\left|W_n(\Theta_r|Y,\varphi)-\int_{\Theta_r}p(\theta^\ast|Y,\varphi)d\theta^\ast\right|.
\end{aligned}
\]
Thus, using Equation 10 from the main text, it is clear that
\[
\lim_{n\rightarrow\infty}\sup_{\theta\in\Theta\setminus\bar{\Theta}_\kappa ,\varphi\in\Phi}\left|p_n^{(\kappa)}(\theta|Y,\varphi)-p^{(\kappa)}(\theta|Y,\varphi)\right|=0.
\]
Since $\mu(\bar{\Theta}_\kappa)=0$, we are done.

\subsection{Proof of Theorem 1}
The theorem naturally holds when $n=1$, we consider the case when $n\geq 2$. Since the dimension $d$ and precision parameter $\kappa$ are known and fixed, we suppose that the parameter space $\Theta$ is equally partitioned and the total number of $d$-orthotopes is $R_\kappa$ and each orthotope is indexed as $\Theta_r$, $r=1,...,R_\kappa$. Since we suppose that the auxiliary chain has converged before we start collecting auxiliary variable $\tilde{\theta}$, by equation 6 in the main text, we could write the probability of the original proposal distribution $P_n^\ast$ taking a value in each partition component $\Theta_r$ as the integral with respect to the target distribution $p(\theta|Y,\varphi)$:
\[
W_\infty(\Theta_r|Y,\varphi) = \int_{\Theta_r} p(\theta|Y,\varphi)d\theta,\ \ r=1,...,,R_\kappa.
\]

Now we define binary random variables $I_r$, $r=1,...,,R_\kappa$ as:
\begin{equation}
  I_r =
    \begin{cases}
      1 & \text{if orthotope $r$ is never visited by auxiliary variables $\tilde{\theta_i}$, $i=1,...,n$;}\\
      0 & \text{otherwise,}
    \end{cases}       
\end{equation}
We then have the expected number of orthotope visited is
\[
\mathbb{E}\left(|\tilde{\Theta}_n^{(\kappa)}|\right) = \mathbb{E} \left( R_\kappa - \sum_{r=1}^{R_\kappa} I_r \right) = R_\kappa - \sum_{r=1}^{R_\kappa} \left(1-W_\infty(\Theta_r|Y,\varphi)\right)^n.
\]
By the method of the Lagrange multipliers, we write the Lagrange function as:
\[
\mathcal{L}(W_\infty(\Theta_1|Y,\varphi),...,W_\infty(\Theta_{R_\kappa}|Y,\varphi),\lambda) = R_\kappa - \sum_{r=1}^{R_\kappa} \left(1-W_\infty(\Theta_r|Y,\varphi)\right)^n + \lambda\left(\sum_{r=1}^{R_\kappa} W_\infty(\Theta_r|Y,\varphi) -1 \right).
\]
Conduct first order partial derivatives, we have
\[
\begin{aligned}
&\frac{\partial \mathcal{L}}{\partial W_\infty(\Theta_r|Y,\varphi)} = -n\left(1-W_\infty(\Theta_r|Y,\varphi)\right)^{n-1} + \lambda =0,  \ \ r=1,...,R_\kappa;\\
&\frac{\partial \mathcal{L}}{\partial \lambda} = \sum_{r=1}^{R_\kappa} W_\infty(\Theta_r|Y,\varphi) - 1 = 0.
\end{aligned}
\]
These equations hold when $W_\infty(\Theta_r|Y,\varphi)=1/R_\kappa$, $r=1,...,R_\kappa$. We now consider the second order derivatives, we have
\[
\begin{aligned}
&\frac{\partial^2 \mathcal{L}}{\partial^2 W_\infty(\Theta_r|Y,\varphi)} = -n(n-1)\left(1-W_\infty(\Theta_r|Y,\varphi)\right)^{n-2}, \ \ r=1,...,R_\kappa;\\
&\frac{\partial^2 \mathcal{L}}{\partial W_\infty(\Theta_r|Y,\varphi) \partial W_\infty(\Theta_t|Y,\varphi)} = 0,\ \ r\neq t.
\end{aligned}
\]
Hence the Hessian matrix is negative definite, $\mathbb{E}\left(|\tilde{\Theta}_n^{(\kappa)}|\right)$ achieves its maxima when $W_\infty(\Theta_r|Y,\varphi)=1/R_\kappa$, $r=1,...,R_\kappa$. If we additionally require this to be held for any precision parameter $\kappa$, the target distribution $p(\theta|Y,\varphi)$ has to be uniform distribution.

\subsection{Proof of Lemma 2}
Given a $(\theta,\varphi)\in\Theta \times\Phi$, for any Borel set $\mathcal{B}=\mathcal{B}_\Theta\times\mathcal{B}_\Phi\subset\Theta \times\Phi$, define a signed measure $D_n$ on $\Theta \times\Phi$ as
\[
\begin{aligned}
&D_n(\mathcal{B}|(\theta,\varphi))=\textbf{T}_n^{(1)}\left(\mathcal{B}|(\theta,\varphi),\mathcal{G}_n\right)-\textbf{U}^{(1)}(\mathcal{B}|(\theta,\varphi)) \\
&=\int_{\mathcal{B}_\Phi}\int_{\mathcal{B}_\Theta}\left( \alpha(\varphi^\prime|\varphi)p^{(\kappa)}(\theta^\prime|Y,\varphi^\prime)q(\varphi^\prime|\varphi)-\alpha(\varphi^\prime|\varphi)p_n^{(\kappa)}(\theta^\prime|Y,\varphi^\prime)q(\varphi^\prime|\varphi)\right) d\theta^\prime d\varphi^\prime \\
&=\int_{\mathcal{B}_\Phi} \left(\int_{\mathcal{B}_\Theta}\left(p^{(\kappa)}(\theta^\prime|Y,\varphi^\prime)-p_n^{(\kappa)}(\theta^\prime|Y,\varphi^\prime)\right)d\theta^\prime\right)\alpha(\varphi^\prime|\varphi)q(\varphi^\prime|\varphi)d\varphi^\prime.
\end{aligned}
\]
Since $p(\varphi|Z)$ and $q(\varphi^\prime|\varphi)$ are continuous on a compact set, then $\alpha(\varphi^\prime|\varphi)$ and $q(\varphi^\prime|\varphi)$ are bounded. Let $C=\sup_{\varphi^\prime\in\Phi, \varphi\in\Phi}\alpha(\varphi^\prime|\varphi)q(\varphi^\prime|\varphi)$, we have

\[
\begin{aligned}
&|D_n(\mathcal{B}|(\theta,\varphi))| \\
&=\left|\int_{\mathcal{B}_\Phi} \left(\int_{\mathcal{B}_\Theta\setminus\bar{\Theta}_\kappa}\left(p^{(\kappa)}(\theta^\prime|Y,\varphi^\prime)-p_n^{(\kappa)}(\theta^\prime|Y,\varphi^\prime)\right)d\theta^\prime\right)\alpha(\varphi^\prime|\varphi)q(\varphi^\prime|\varphi)d\varphi^\prime\right| \\
&\leq \int_{\mathcal{B}_\Phi} \sup_{\varphi^\ast\in\Phi}\left|\int_{\mathcal{B}_\Theta\setminus\bar{\Theta}_\kappa}\left(p^{(\kappa)}(\theta^\prime|Y,\varphi^\ast)-p_n^{(\kappa)}(\theta^\prime|Y,\varphi^\ast)\right)d\theta^\prime\right|Cd\varphi^\prime \\
&\leq\mu(\Phi)C\int_{\mathcal{B}_\Theta\setminus\bar{\Theta}_\kappa}\sup_{\theta^\ast\in\Theta\setminus\bar{\Theta}_\kappa ,\varphi^\ast\in\Phi}\left|p^{(\kappa)}(\theta^\ast|Y,\varphi^\ast)-p_n^{(\kappa)}(\theta^\ast|Y,\varphi^\ast)\right|d\theta^\prime \\
&\leq \mu(\Phi)\mu(\Theta)C\sup_{\theta^\ast\in\Theta\setminus\bar{\Theta}_\kappa ,\varphi^\ast\in\Phi}\left|p^{(\kappa)}(\theta^\ast|Y,\varphi^\ast)-p_n^{(\kappa)}(\theta^\ast|Y,\varphi^\ast)\right|.\\
\end{aligned}
\]
The important fact here is that $|D_n(\mathcal{B}|(\theta,\varphi))|$ can be uniformly (with respect to $\theta$, $\varphi$ and Borel set $\mathcal{B}$) bounded by
\[
\sup_{\theta^\ast\in\Theta\setminus\bar{\Theta}_\kappa ,\varphi^\ast\in\Phi}\left|p^{(\kappa)}(\theta^\ast|Y,\varphi^\ast)-p_n^{(\kappa)}(\theta^\ast|Y,\varphi^\ast)\right|
\]
up to a constant. 

Given Lemma 1, we have that the density $p_n^{(\kappa)}$ converges almost surely to $p^{(\kappa)}$ and this convergence is uniformly on $\Theta\setminus\bar{\Theta}_\kappa \times\Phi$, and so we have
\[
\lim_{n\rightarrow\infty}\sup_{\theta\in\Theta ,\varphi\in\Phi}\left\|D_n(\cdot|(\theta,\varphi))\right\|_{TV}=0.
\]
Now by the triangle inequality, we have
\[
\begin{aligned}
&\lim_{n\rightarrow\infty}\sup_{\theta\in\Theta ,\varphi\in\Phi}\left\|\textbf{T}_{n+1}^{(1)}\left(\cdot|(\theta,\varphi),\mathcal{G}_{n+1}\right)-\textbf{T}_n^{(1)}\left(\cdot|(\theta,\varphi),\mathcal{G}_n\right)\right\|_{TV} \\
&\leq \lim_{n\rightarrow\infty}\sup_{\theta\in\Theta ,\varphi\in\Phi}\left\|D_{n+1}(\cdot|(\theta,\varphi))\right\|_{TV}+\lim_{n\rightarrow\infty}\sup_{\theta\in\Theta ,\varphi\in\Phi}\left\|D_n(\cdot|(\theta,\varphi))\right\|_{TV}.
\end{aligned}
\]
It follows that:
\[
\lim_{n\rightarrow\infty}\sup_{\theta\in\Theta ,\varphi\in\Phi}\left\|\textbf{T}_{n+1}^{(1)}\left(\cdot|(\theta,\varphi),\mathcal{G}_{n+1}\right)-\textbf{T}_n^{(1)}\left(\cdot|(\theta,\varphi),\mathcal{G}_n\right)\right\|_{TV}=0.
\]

\subsection{Proof of Lemma 3}
Define a function $g:\Phi\rightarrow\mathbb{R}$ as
\[
g(\varphi)=\min_{\theta\in\Theta } p_n^{(\kappa)}(\theta|Y,\varphi).
\]
Since the support of $p_n^{(\kappa)}$ is $\Theta $, we have $g(\varphi)>0$, for all $\varphi\in\Phi$. In addition, since each element of $\mathscr{W}_n(\varphi)$ is a continuous function on the compact set $\Phi$ (see equation 5 and 9 in the main text), then $g(\varphi)$ is also a continuous function on $\Phi$. Since $\Phi$ is compact, $g(\varphi)$ reaches its minima
\[
\varepsilon=\min_{\varphi\in\Phi}g(\varphi).
\]
Thus $p_n^{(\kappa)}(\theta|Y,\varphi)>\varepsilon$ for all $\theta\in\Theta$ and $\varphi\in\Phi$, and local positivity holds.

By the same reasoning, it is also true for the proposal distribution with density $p^{(\kappa)}$.

\subsection{Necessary definitions}
\begin{definition}
Given any function $V:\Psi\rightarrow [1,\infty)$ and any signed measure $\mathcal{M}$ on $\Psi$, define the $V$-norm as
\[
\left\|\mathcal{M}\right\|_V=\sup_{|g|\leq V}\left|\int_{\Psi} g(\psi)\mathcal{M}(d\psi)\right|.
\]
\end{definition}

\begin{definition}
For simplicity, for any function $f:\Psi\rightarrow\mathbb{R}$ and any measure $\mathcal{M}$ on $\Psi$, write
\[
\mathcal{M}f:=\int_\Psi f(\psi)\mathcal{M}(d\psi).
\]
\end{definition}

\begin{definition}
Given any two measures $\textbf{M}_{(x)}(dz):=\textbf{M}(dz|x)$, where $x\in\mathbb{X}$, and $\textbf{N}_{(y)}(dx):=\textbf{N}(dx|y)$ which concentrates on $\mathbb{X}$, for any Borel set $\mathcal{B}$, we write
\[
\textbf{M}\textbf{N}_{(y)}(\mathcal{B}):=\int_\mathcal{B}\int_\mathbb{X} \textbf{M}_{(x)}(dz)\textbf{N}_{(y)}(dx).
\]
The definition can be extended to cases with more than two measures in a natural way.
\end{definition}

\subsection{Proof of Lemma 4}
Given the filtration $\mathcal{G}_n$, the transition kernel $\textbf{U}^{(1)}$ and $\textbf{V}_n^{(1)}$ both admit an irreducible and aperiodic Markov chain by assumption. Therefore, to prove that transition kernel $\textbf{T}_n^{(1)}$ also holds same property, it suffices to prove that for any $s\in\mathbb{N}$, $(\theta_0,\varphi_0)\in\Theta \times\Phi$, and Borel set $\mathcal{B}=\mathcal{B}_\Theta\times\mathcal{B}_\Phi\subset\Theta \times\Phi$ such that $\textbf{V}_n^{(s)}(\mathcal{B})>0$, we have $\textbf{T}_n^{(s)}(\mathcal{B})>0$. We prove this by mathematical induction. 

Consider first when $s=1$. We write $\alpha(\varphi^\prime|\varphi)=\min(1,\beta(\varphi^\prime|\varphi))$ where
\[
\beta(\varphi^\prime|\varphi)=\frac{p(\varphi^\prime|Z)q(\varphi|\varphi^\prime)}{p(\varphi|Z)q(\varphi^\prime|\varphi)},
\]
and $\alpha_n((\theta^\prime,\varphi^\prime)|(\theta,\varphi))=\min\left(1,\beta_n((\theta^\prime,\varphi^\prime)|(\theta,\varphi))\right)$, where
\[
\beta_n((\theta^\prime,\varphi^\prime)|(\theta,\varphi))=\frac{p^{(\kappa)}(\theta^\prime|Y,\varphi^\prime)p(\varphi|Z)q(\varphi|\varphi^\prime)p_n^{(\kappa)}(\theta|Y,\varphi)}{p^{(\kappa)}(\theta|Y,\varphi)p(\varphi|Z)q(\varphi^\prime|\varphi)p_n^{(\kappa)}(\theta^\prime|Y,\varphi^\prime)},
\]
and
\[
r((\theta^\prime,\varphi^\prime),(\theta,\varphi))=\frac{\beta(\varphi^\prime|\varphi)}{\beta_n((\theta^\prime,\varphi^\prime)|(\theta,\varphi))},
\]
noting that both $p_n^{(\kappa)}$ and $p^{(k)}$ are bounded away from $0$ and $\infty$. Now we denote 
\[
r^\ast=\min_{(\theta^\prime,\varphi^\prime),(\theta,\varphi)\in\Theta \times\Phi}r((\theta^\prime,\varphi^\prime),(\theta,\varphi)),
\]
and it is easy to see that $r^\ast>0$.

Now given any Borel set $\mathcal{B}=\mathcal{B}_\Theta\times\mathcal{B}_\Phi\subset\Theta \times\Phi$ and initial value $(\theta_0,\varphi_0)\in\Theta \times\Phi$, we have
\[
\resizebox{.95\hsize}{!}{$\begin{aligned}
&\textbf{T}_n^{(1)}\left(\mathcal{B}|(\theta_0,\varphi_0),\mathcal{G}_n\right) \\
&=\textbf{T}_n^{(1)}\left(\mathcal{B}\setminus \{(\theta_0,\varphi_0)\}|(\theta_0,\varphi_0),\mathcal{G}_n\right)\\
&=\int_\mathcal{B}\alpha(\varphi|\varphi_0) p_n^{(\kappa)}(\theta|Y,\varphi)q(\varphi|\varphi_0)d\theta d\varphi \\
&=\int_\mathcal{B}\min\left\lbrace 1, r((\theta,\varphi),(\theta_0,\varphi_0))\beta_n((\theta,\varphi)|(\theta_0,\varphi_0)) \right\rbrace p_n^{(\kappa)}(\theta|Y,\varphi)q(\varphi|\varphi_0)d\theta d\varphi \\
&\geq \int_\mathcal{B}\min\left\lbrace1, r((\theta,\varphi),(\theta_0,\varphi_0))\right\rbrace\min\left\lbrace 1, \beta_n((\theta,\varphi)|(\theta_0,\varphi_0)) \right\rbrace p_n^{(\kappa)}(\theta|Y,\varphi)q(\varphi|\varphi_0)d\theta d\varphi \\
&\geq \min\left\lbrace1, r^\ast\right\rbrace\int_\mathcal{B}\alpha_n((\theta,\varphi)|(\theta_0,\varphi_0)) p_n^{(\kappa)}(\theta|Y,\varphi)q(\varphi|\varphi_0)d\theta d\varphi \\
\end{aligned}$}
\]
Since $\min\left\lbrace1, r^\ast\right\rbrace>0$, we have 
\[
\textbf{V}_n^{(1)}\left(\mathcal{B}|(\theta_0,\varphi_0),\mathcal{G}_n\right)>0 \Rightarrow \textbf{T}_n^{(1)}\left(\mathcal{B}|(\theta_0,\varphi_0),\mathcal{G}_n\right)>0.
\]
Thus, the induction assumption holds when $s=1$.

Now assume that the induction assumption holds up to step $s=s^\ast$, i.e.
\[
\textbf{V}_n^{(s^\ast)}(\mathcal{B})>0 \Rightarrow\textbf{T}_n^{(s^\ast)}(\mathcal{B})>0.
\]
We need to show that it also holds at step $s=s^\ast+1$. For an initial value $(\theta_0,\varphi_0)$, consider a Borel set $\mathcal{B}$ such that $\textbf{V}_n^{(s^\ast+1)}(\mathcal{B})>0$. We proceed by contradiction. Suppose that
\[
\textbf{T}_n^{(s^\ast+1)}(\mathcal{B})=\int_{\Theta \times\Phi} \textbf{T}_n^{(1)}\left(\mathcal{B}|(\theta,\varphi),\mathcal{G}_n\right)\textbf{T}_n^{(s^\ast)}(d\theta,d\varphi)=0.
\]
This implies that the function $\textbf{T}_n^{(1)}\left(\mathcal{B}|\cdot,\mathcal{G}_n\right)=0$ almost surely with respect to the measure $\textbf{T}_n^{(s^\ast)}$. Because the induction assumption holds at step $s^\ast$, which means that any $\textbf{V}_n^{(s^\ast)}$-measurable set of positive measure is a subset of a $\textbf{T}_n^{(s^\ast)}$-measurable set of positive measure, we have that the function $\textbf{T}_n^{(1)}\left(\mathcal{B}|\cdot,\mathcal{G}_n\right)=0$ almost surely with respect to the measure $\textbf{V}_n^{(s^\ast)}$. This further implies that the function $\textbf{V}_n^{(1)}\left(\mathcal{B}|\cdot,\mathcal{G}_n\right)=0$ almost surely with respect to the measure $\textbf{V}_n^{(s^\ast)}$. It is clear that this contradicts the fact that $\textbf{V}_n^{(s^\ast+1)}(\mathcal{B})>0$. Hence, we are done.

Given that $q(\varphi^\prime|\varphi)$ and $p^{(\kappa)}(\theta^\prime|Y,\varphi^\prime)$ satisfy the local positivity by Lemma 3, it is easy to check that
\[
q((\theta^\prime,\varphi^\prime)|(\theta,\varphi))=p^{(\kappa)}(\theta^\prime|Y,\varphi^\prime)q(\varphi^\prime|\varphi)
\]
also satisfies local positivity. Hence, by Theorem 2.2 of \cite{10.2307/2337435}, since the target distribution is bounded away from 0 and $\infty$ on a compact set and the proposal distribution satisfies local positivity, the Partial Gibbs chain is irreducible and aperiodic, and every nonempty compact set is small. Moreover, $\Theta\times\Phi$ is a small set for the transition kernel $\textbf{u}^{(1)}(\cdot|(\theta,\varphi))$, since it is compact. Hence, it is straightforward to verify that, for any $(\theta,\varphi)\in\Theta\times\Phi$ and Borel set $\mathcal{B}\subset\Theta\times\Phi$, there exists a $\delta>0$ such that
\[
\textbf{U}^{(1)}(\mathcal{B}|(\theta,\varphi))\geq \delta\mu(\mathcal{B}).
\]

Since 
\[
q_n((\theta^\prime,\varphi^\prime)|(\theta,\varphi))=p_n^{(\kappa)}(\theta^\prime|Y,\varphi^\prime)q(\varphi^\prime|\varphi)
\]
also satisfies local positivity, following the proof of Theorem 2.2 in \cite{10.2307/2337435} \footnote{The difference is that there is an additional term, the ratio of $p_n^{(\kappa)}$ to $p^{(\kappa)}$, in our case. Since they are positive and bounded functions defined on $\Theta\times\Phi$, this ratio has a positive minimum on $\Theta\times\Phi$. Hence, the inequality in the original proof still holds.}, one can show that, $\Theta\times\Phi$ is also a small set for the transition kernel $\textbf{t}_n^{(1)}$. Let the ``geometric drift function'' $V(\theta,\varphi)\equiv1$, there exists $\lambda<1$ and $b<\infty$ such that
\[
1=\int_{\Theta\times\Phi} V(\theta^\ast,\varphi^\ast)\textbf{T}_n^{(1)}\left((d\theta^\ast,d\varphi^\ast)|(\theta,\varphi),\mathcal{G}_n\right)\leq \lambda V(\theta,\varphi)+b\mathbbm{1}_{\{(\theta,\varphi)\in\Theta\times\Phi\}}
\]
then by Theorem 3.1 of \cite{10.2307/2337435}, for all $(\theta_0,\varphi_0)\in\Theta\times\Phi$, there exists a probability measure $\Pi_n$ on $\Theta \times\Phi$ and constant $\rho<1$ and $R<\infty$ such that for all $s=1,2,...$ and all $(\theta_0,\varphi_0)\in\Theta\times\Phi$,
\[
\left\|\textbf{T}_n^{(s)}-\Pi_n\right\|_V\leq R\ V(\theta_0,\varphi_0)\rho^s.
\]
Since $V=1$, we have uniformly geometric convergence:
\[
\lim_{s\rightarrow\infty}\sup_{(\theta_0,\varphi_0)\in\Theta\times\Phi} \left\|\textbf{T}_n^{(s)}-\Pi_n\right\|_V=0
\]
In addition, for any $(\theta_0,\varphi_0)\in\Theta\times\Phi$,
\[
0\leq \left\|\textbf{T}_n^{(s)}(\cdot)-\Pi_n\left(\cdot\right)\right\|_{TV}\leq \left\|\textbf{T}_n^{(s)}-\Pi_n\right\|_V,
\]
by the squeeze theorem, we have:
\[
\lim_{s\rightarrow\infty}\sup_{(\theta_0,\varphi_0)\in\Theta\times\Phi} \left\|\textbf{T}_n^{(s)}(\cdot)-\Pi_n\left(\cdot\right)\right\|_{TV}=0.
\]
\begin{remark}
Following the fact that, for any $(\theta,\varphi)\in\Theta\times\Phi$ and Borel set $\mathcal{B}\subset\Theta\times\Phi$, there exists a $\delta>0$ such that
\[
\textbf{U}^{(1)}(\mathcal{B}|(\theta,\varphi))\geq \delta\mu(\mathcal{B}).
\]
following the proof of Lemma 2, we have:
\[
\begin{aligned}
&\textbf{U}^{(1)}(\mathcal{B}|(\theta,\varphi))=\textbf{U}^{(1)}(\mathcal{B}|(\theta,\varphi))-\textbf{T}_n^{(1)}\left(\mathcal{B}|(\theta,\varphi),\mathcal{G}_n\right)+\textbf{T}_n^{(1)}\left(\mathcal{B}|(\theta,\varphi),\mathcal{G}_n\right) \\
&\leq \sup_{\theta\in\Theta, \varphi\in\Phi}\left|\textbf{U}^{(1)}(\mathcal{B}|(\theta,\varphi))-\textbf{T}_n^{(1)}\left(\mathcal{B}|(\theta,\varphi),\mathcal{G}_n\right)\right|+\textbf{T}_n^{(1)}\left(\mathcal{B}|(\theta,\varphi),\mathcal{G}_n\right)\\
&\leq C\mu(\mathcal{B})\sup_{\theta^\ast\in\Theta\setminus\bar{\Theta}_\kappa ,\varphi^\ast\in\Phi}\left|p^{(\kappa)}(\theta^\ast|Y,\varphi^\ast)-p_n^{(\kappa)}(\theta^\ast|Y,\varphi^\ast)\right|+\textbf{T}_n^{(1)}\left(\mathcal{B}|(\theta,\varphi),\mathcal{G}_n\right). \\
\end{aligned}
\]
where $C$ is a constant. Therefore, for any $(\theta,\varphi)\in\Theta\times\Phi$ and Borel set $\mathcal{B}\subset\Theta\times\Phi$, we have
\[
\textbf{T}_n^{(1)}\left(\mathcal{B}|(\theta,\varphi),\mathcal{G}_n\right)\geq \left(\delta-C\sup_{\theta^\ast\in\Theta\setminus\bar{\Theta}_\kappa ,\varphi^\ast\in\Phi}\left|p^{(\kappa)}(\theta^\ast|Y,\varphi^\ast)-p_n^{(\kappa)}(\theta^\ast|Y,\varphi^\ast)\right|\right)\mu(\mathcal{B}).
\]
Note that, by Lemma 1, for any outcome $\omega$ in probability space $\Omega$, we have
\[
\sup_{\theta^\ast\in\Theta\setminus\bar{\Theta}_\kappa ,\varphi^\ast\in\Phi}\left|p^{(\kappa)}(\theta^\ast|Y,\varphi^\ast)-p_n^{(\kappa)}(\theta^\ast|Y,\varphi^\ast)\right|\rightarrow 0,\ \text{when}\ n\rightarrow\infty.
\]
This is important. Since for any positive constant $a<\delta$, there exists a $N$ such that for all $n>N$, we have
\[
\textbf{T}_n^{(1)}\left(\mathcal{B}|(\theta,\varphi),\mathcal{G}_n\right)\geq (\delta-a)\mu(\mathcal{B}).
\]
Hence, a common and same lower bound is well defined on this outcome $\omega$.
\end{remark}

\subsection{Proof of Lemma 5}
For any initial value $(\theta_0,\varphi_0)$ and $s>1$ and function $f:\Theta\times\Phi\rightarrow [-1,1]$, write
\[
\textbf{T}_n^{(s)}f-P_{cut}^{(\kappa)}f= \textbf{U}^{(s)}f- P_{cut}^{(\kappa)}f + \textbf{T}_n^{(s)}f- \textbf{U}^{(s)}f.
\]
We first concentrate on the second term $\textbf{T}_n^{(s)}f- \textbf{U}^{(s)}f$, for any $1\leq s_0<s$, denote $\textbf{U}^{(0)}=1$ and $\textbf{T}_n^{(0)}=1$, we have, by a telescoping argument,
\[
\resizebox{.95\hsize}{!}{$\begin{aligned}
&\left|\textbf{T}_n^{(s)}f- \textbf{U}^{(s)}f\right| \\
&\leq \left|\textbf{T}_n^{(s)}f-\textbf{T}_n^{(s_0)}f\right|+\left|\textbf{T}_n^{(s_0)}f-\textbf{U}^{(s_0)}f\right|+\left|\textbf{U}^{(s)}f-\textbf{U}^{(s_0)}f\right| \\
&=\left|\textbf{T}_n^{(s)}f-\textbf{T}_n^{(s_0)}f\right|+\left|\sum_{k=0}^{s_0-1}\left(\textbf{U}^{(k)}\textbf{T}_n^{(s_0-k)}f-\textbf{U}^{(k+1)}\textbf{T}_n^{(s_0-k-1)}f\right)\right|+\left|\textbf{U}^{(s)}f-\textbf{U}^{(s_0)}f\right| \\
&=\left|\textbf{T}_n^{(s)}f-\textbf{T}_n^{(s_0)}f\right|+ \left|\sum_{k=0}^{s_0-1}\textbf{U}^{(k)}\left(\textbf{T}_n^{(1)}-\textbf{U}^{(1)}\right)\textbf{T}_n^{(s_0-k-1)}f\right| + \left|\textbf{U}^{(s)}f-\textbf{U}^{(s_0)}f\right|.\\
\end{aligned}$}
\]
Note that, $\left(\textbf{T}_n^{(1)}-\textbf{U}^{(1)}\right)$ is the signed measure $D_n$ defined in the proof of Lemma 2. By the result of Lemma 2, we have
\[
\sup_{\theta\in\Theta ,\varphi\in\Phi}\left\|D_n(\cdot|(\theta,\varphi))\right\|_{TV}\xrightarrow{\text{a.s.}}0,
\]
on the probability space $\Omega$. Then by Egorov's theorem, for any $e>0$, there exists a set $E_1\subset\Omega$ with $\mathbb{P}(E_1)>1-\frac{e}{2}$ such that $\sup_{\theta\in\Theta ,\varphi\in\Phi}\left\|D_n(\cdot|(\theta,\varphi))\right\|_{TV}$ uniformly converges to 0 on $E_1$. Hence, for any $\epsilon>0$, there exists a $N_1(\epsilon)$, such that for all $n>N_1(\epsilon)$, $\sup_{\theta\in\Theta ,\varphi\in\Phi}\left\|D_n(\cdot|(\theta,\varphi))\right\|_{TV}\leq\epsilon$ on $E_1$. Then, since the remaining terms are bounded by 1, there exist a constant $C$ such that
\[
\left|\sum_{k=0}^{s_0-1}\textbf{U}^{(k)}\left(\textbf{T}_n^{(1)}-\textbf{U}^{(1)}\right)\textbf{T}_n^{(s_0-k-1)}f\right|\leq C s_0 \epsilon.
\]

Now, following the same reasoning as Lemma 4 and Theorem 3.1 of \cite{10.2307/2337435}, $\textbf{U}^{(s)}$ uniformly converges to $P_{cut}^{(\kappa)}$ in the sense of $V$-norm ($V\equiv 1$). Hence, for the same $\epsilon$, there exists a $S_1(\epsilon)$ such that for any $s>s_0>S_1(\epsilon)$,
\[
\left|\textbf{U}^{(s)}f-\textbf{U}^{(s_0)}f\right|\leq \epsilon,\quad \left|\textbf{U}^{(s)}f- P_{cut}^{(\kappa)}f\right|\leq \epsilon.
\]

By Lemma 1, we have that $p_n^{(\kappa)}(\theta|Y,\varphi)$ converges to $p^{(\kappa)}(\theta|Y,\varphi)$ almost surely on probability space $\Omega$. Then by Egorov's theorem, for same $e$, there exists a set $E_2\subset\Omega$ with $\mathbb{P}(E_2)>1-\frac{e}{2}$ such that $p_n^{(\kappa)}(\theta|Y,\varphi)$ uniformly converges to $p^{(\kappa)}(\theta|Y,\varphi)$ on $E_2$. Hence on $E_2$, by the Remark of the proof of Lemma 4, for any Borel set $\mathcal{B}\subset\Theta\times\Phi$ and $(\theta,\varphi)\in\Theta\times\Phi$, there exists a $N_2$ such that for all $n>N_2$,
\[
\textbf{T}_n^{(1)}\left(\mathcal{B}|(\theta,\varphi),\mathcal{G}_n\right)\geq \frac{\delta}{2}\mu(\mathcal{B}).
\]
By Theorem 2.3 of \cite{10.2307/2245077}, we have all $\textbf{T}_n^{(1)}\left(\cdot|(\theta,\varphi),\mathcal{G}_n\right)$, when $n>N_2$, are uniformly ergodic in $V$-norm and have the same geometric convergence rate. Hence on $E_2$, there exists a $S_2(\epsilon)$, such that for all $s>s_0>S_2(\epsilon)$ and $n>N_2$,
\[
\left|\textbf{T}_n^{(s)}f-\textbf{T}_n^{(s_0)}f\right|\leq \epsilon.
\]

Let $N(\epsilon)=\max(N_1(\epsilon),N_2)$ and $S(\epsilon)=\max(S_1(\epsilon),S_2(\epsilon))$. On set $E_2$, all convergences which involve $S_1(\epsilon)$ and $S_2(\epsilon)$ have geometric convergence rate. Thus, one can select a $S(\epsilon)$ such that $\epsilon S(\epsilon)\rightarrow 0$ when $\epsilon\rightarrow 0$.

Let $\varepsilon=(C S(\epsilon)+3)\epsilon$ and set $E=E_1\cap E_2$ with $\mathbb{P}(E)>1-e$. It is clear that $\varepsilon\rightarrow 0$ when $\epsilon\rightarrow 0$. We can conclude that, on set $E$, there exists $N(\epsilon)$ and $S(\epsilon)$ such that for any $n>N(\epsilon)$ and $s>S(\epsilon)$,
\[
\left|\textbf{T}_n^{(s)}f-P_{cut}^{(\kappa)}f\right|\leq \varepsilon.
\]
Note that, for any Borel set $\mathcal{B}\subset\Theta\times\Phi$, we can let function $f$ be an indicator function $\mathbbm{1}_{\{x\in\mathcal{B}\}}$. Hence, for any initial value $(\theta_0,\varphi_0)\in\Theta\times\Phi$, and any $\varepsilon>0$ and $e>0$, there exists constants $S(\varepsilon)>0$ and $N(\varepsilon)>0$ such that
\[
\mathbb{P}\left(\left\lbrace P_n^{(\kappa)}:\ \left\|\textbf{T}_n^{(s)}\left(\cdot\right)-P_{cut}^{(\kappa)}\left(\cdot\right)\right\|_{TV}\leq\varepsilon\right\rbrace\right)>1-e.
\]
for all $s>S(\varepsilon)$ and $n>N(\varepsilon)$.

\subsection{Proof of Corollary 2}
Given the result of global convergence in Corollary 1, and given a $\varphi$, there is a subset $\Theta^\ast\subset\Theta$ such that
\[
\sup_{\theta\in\Theta^\ast}\left|p(\theta|Y,\varphi)-p^{(\kappa)}(\theta|Y,\varphi)\right|\leq \sup_{\theta\in\Theta^\ast}\left\|\nabla_{\theta}p(\theta|Y,\varphi)\right\|_2 \frac{\sqrt{d}}{10^\kappa},
\]
where $d$ is the dimension of $\theta$. Following the proof of Lemma 1, we know that the construction of the set $\Theta^\ast$ is only related to the geometric shape of $\Theta$, and it is not related to the function and thus not related to $\varphi$. Since $p_{cut}$ is continuously differentiable, then $\nabla_{\theta,\varphi}p_{cut}(\theta,\varphi)$ is continuous. This further implies $\nabla_{\theta}p(\theta|Y,\varphi)$ is continuous with respect to $\theta$ and $\varphi$. Because $\Phi$ is compact, we have
\[
\sup_{\theta\in\Theta^\ast,\varphi\in\Phi}\left|p(\theta|Y,\varphi)-p^{(\kappa)}(\theta|Y,\varphi)\right|\leq \sup_{\theta\in\Theta^\ast,\varphi\in\Phi}\left\|\nabla_{\theta}p(\theta|Y,\varphi)\right\|_2 \frac{\sqrt{d}}{10^\kappa}<\infty.
\]

Now since $\mu(\Theta^\ast)=\mu(\Theta)$, we have the following bias term
\[
\begin{aligned}
&\left|\int_{\Theta\times\Phi}f(\theta,\varphi)P_{cut}(d\theta,d\varphi)-\int_{\Theta\times\Phi}f(\theta,\varphi)P_{cut}^{(\kappa)}(d\theta,d\varphi)\right| \\
&=\left|\int_{\Theta^\ast\times\Phi} f(\theta,\varphi)\left(p(\theta|Y,\varphi)-p^{(\kappa)}(\theta|Y,\varphi)\right)p(\varphi|Z) d\theta d\varphi\right| \\
&\leq \int_{\Theta^\ast\times\Phi} f(\theta,\varphi)\left|p(\theta|Y,\varphi)-p^{(\kappa)}(\theta|Y,\varphi)\right|p(\varphi|Z) d\theta d\varphi \\
&\leq \sup_{\theta\in\Theta^\ast,\varphi\in\Phi}\left|p(\theta|Y,\varphi)-p^{(\kappa)}(\theta|Y,\varphi)\right|\int_{\Theta^\ast\times\Phi} f(\theta,\varphi)p(\varphi|Z) d\theta d\varphi \\
&\leq \sup_{\theta\in\Theta^\ast,\varphi\in\Phi}\left\|\nabla_{\theta}p(\theta|Y,\varphi)\right\|_2 \frac{\sqrt{d}}{10^\kappa}\left(\int_{\Theta^\ast\times\Phi} f(\theta,\varphi)p(\varphi|Z) d\theta d\varphi\right).
\end{aligned}
\]
For any $\varepsilon>0$, let
\[
\sup_{\theta\in\Theta^\ast,\varphi\in\Phi}\left\|\nabla_{\theta}p(\theta|Y,\varphi)\right\|_2 \frac{\sqrt{d}}{10^\kappa}\left(\int_{\Theta^\ast\times\Phi} f(\theta,\varphi)p(\varphi|Z) d\theta d\varphi\right)=\frac{\varepsilon}{2},
\]
let the solution of this equation be $\kappa^\ast$. We have the following bias term
\[
\left|\int_{\Theta\times\Phi}f(\theta,\varphi)P_{cut}(d\theta,d\varphi)-\int_{\Theta\times\Phi}f(\theta,\varphi)P_{cut}^{(\kappa^\ast)}(d\theta,d\varphi)\right|\leq \frac{\varepsilon}{2},
\]
and this is always true in probability space $\Omega$. Now by Theorem 2, for the same $\varepsilon$ and $\kappa^\ast$, there exists a $N(\kappa^\ast,\varepsilon)$ such that for any $N>N(\kappa^\ast,\varepsilon)$, there is a set $E\subset\Omega$ with $\mathbb{P}(E)>1-e$ and on this set the error term satisfies
\[
\left|\frac{1}{N}\sum_{n=1}^N f(\theta_n,\varphi_n) - \int_{\Theta\times\Phi} f(\theta,\varphi)P_{cut}^{(\kappa^\ast)}(d\theta,d\varphi)\right|\leq \frac{\varepsilon}{2}.
\]
Hence, combining the error term and bias term, on the set $E$ we have
\[
\resizebox{.95\hsize}{!}{$\begin{aligned}
&\left|\frac{1}{N}\sum_{n=1}^N f(\theta_n,\varphi_n)-\int_{\Theta\times\Phi} f(\theta,\varphi)P_{cut}(d\theta,d\varphi)\right| \\
&\leq \left|\frac{1}{N}\sum_{n=1}^N f(\theta_n,\varphi_n) - \int_{\Theta\times\Phi} f(\theta,\varphi)P_{cut}^{(\kappa^\ast)}(d\theta,d\varphi)\right|+ \left|\int_{\Theta\times\Phi}f(\theta,\varphi)P_{cut}(d\theta,d\varphi)-\int_{\Theta\times\Phi}f(\theta,\varphi)P_{cut}^{(\kappa^\ast)}(d\theta,d\varphi)\right|\leq \varepsilon 
\end{aligned}$}
\]
Hence, we are done.

\end{document}